%% file: mpfc_arXiv.tex
\documentclass[9pt,twocolumn]{IEEEtran}
\ifCLASSOPTIONcompsoc
   \usepackage[nocompress]{cite}
\else
   \usepackage[sort,compress]{cite}
\fi
\usepackage{graphics, graphicx}

\ifCLASSINFOpdf
   \usepackage{graphics, graphicx, subfig}
\else
\fi
\usepackage[cmex10]{amsmath}
\usepackage{amssymb}
\usepackage{algorithmic}

\usepackage{mdwmath}
\usepackage{mdwtab}
\ifCLASSOPTIONcompsoc
  \usepackage[caption=false,font=normalsize,labelfont=sf,textfont=sf]{subfig}
\else
  \usepackage{subfig}
\fi

\usepackage{color}

 \newtheorem{rema}{Remark}
 \newtheorem{lemm}{Lemma}
 \newtheorem{prop}{Proposition}
\newtheorem{thm}{Theorem}
\newtheorem{prob}{Problem}

\newtheorem{ass}{Assumption}

\newcommand{\diag}{\operatorname{diag}}
\newcommand{\interior}{{\operatorname{int}}}

\newcommand{\inte}[1]{\interior{ (\mcl{#1})}}
\newcommand{\minimize}{\operatorname{minimize~}} 
\newcommand{\maximize}{\operatorname{maximize~}}

\newcommand{\mbb}[1]{\mathbb #1}
\newcommand{\mbf}[1]{\mathbf #1} 
\newcommand{\mcl}[1]{\mathcal #1}
\newcommand{\argmin}[1]{\underset{#1}{argmin~}}

\newcommand{\Rnx}{\mbb{R}^{n_x}}
\newcommand{\Rny}{\mbb{R}^{n_y}}
\newcommand{\Rnu}{\mbb{R}^{n_u}}
\newcommand{\Rnz}{\mbb{R}^{\hat r}}

\DeclareMathOperator{\Lie}{L}
\newcommand{\Lfhx}[2]{\Lie_f^{#1}h_{#2}(x)}

\newcommand{\LomegaLphipsi}[3]{\Lie_{\omega_{#1}}\Lie_{\phi}^{#2}\psi_{#3}(\chi)}

\begin{document}
%
\title{Nonlinear Model Predictive Control for \\
Constrained Output Path Following}
%
%
%
%

\author{\normalsize Timm~Faulwasser  \\ \small{Laboratoire d'Automatique, Ecole Polytechnique F\'ed\'erale de Lausanne, \\ EPFL STI IGM LA
Station 9, CH-1015 Lausanne, Switzerland. \\ Fon: $+$41 21 69 37341 Fax: $+$41 21 69 32574. \quad E-mail: timm.faulwasser@epfl.ch} (correspondence address)\\~\\
\normalsize Rolf Findeisen \\ \small Institute for Automation Engineering, Otto von Guericke University Magdeburg, \\
Universit\"atsplatz 2, 39106 Magdeburg, Germany. \\
Fon: $+$49 391 67 18577 Fax: $+$49 391 67 11191. \quad E-mail: rolf.findeisen@ovgu.de. \normalsize
\IEEEcompsocitemizethanks{Main parts of this research have been conducted while TF was with the Institute for Automation Engineering, Otto-von-Guericke-University  Magdeburg, Germany. }

}

%
%

\markboth{Preprint}%
{Shell \MakeLowercase{\textit{et al.}}: Bare Demo of IEEEtran.cls for Computer Society Journals}
%



\IEEEcompsoctitleabstractindextext{%
\begin{abstract}
We consider the tracking of geometric paths in output spaces of nonlinear systems subject to input and state constraints without pre-specified timing requirements. 
Such problems are commonly referred to as constrained output path-following problems. 
Specifically, we propose a predictive control approach to constrained path-following problems with and without velocity assignments and provide sufficient convergence conditions based on terminal regions and end penalties. Furthermore, we analyze the geometric nature of constrained output path-following problems and thereby provide insight into the computation of suitable terminal control laws and terminal regions. We draw upon an example from robotics to illustrate our findings. 
\end{abstract}

}

\maketitle

\IEEEdisplaynotcompsoctitleabstractindextext

%
\IEEEpeerreviewmaketitle

\input{introduction}

\input{problem_statement}

\input{mpfc}

\input{term}

\input{example}
\input{conclusion}

\appendices
\input{appendix}

\ifCLASSOPTIONcompsoc
  \section*{Acknowledgments}
\else
  \section*{Acknowledgment}
\fi

The authors would like to thank Friedrich von Haeseler from the Otto-von-Guericke University Magdeburg for valuable feedback and discussions on transverse normal forms.

\ifCLASSOPTIONcaptionsoff
  \newpage
\fi



%
\bibliographystyle{IEEEtran}
\bibliography{literature_latin1}

\end{document}

%% file: introduction.tex
\section{Introduction}
The prototypical problem in control is the stabilization of a set-point. Besides stabilization, the design of controllers for the tracking of time-varying references is also well-understood. 
Yet not all problems encountered in applications are set-point stabilization or trajectory-tracking problems. One example is the precise steering of a robotic tool along a geometric curve in the robot workspace. Typically, the highest priority is given to the minimization of  the deviation between the geometric reference path and the robot tool. The velocity to move along the reference is of secondary interest and might be adjusted in order to achieve better accuracy. 
Thus neither the stabilization of a set-point nor the tracking of a pre-defined time-varying reference is at the core of this problem. Such control problems---that require to steer a system along a geometric reference curve, whereby the speed along this reference is a degree of freedom in the controller design---are termed \textit{path-following problems} \cite{Aguiar05a,Dacic06,Do04}.  

 Besides its relevance for applications recent interest in path following is motivated by the fact that in contrast to trajectory tracking, path-following tasks of non-minimum-phase systems are not necessarily subject to fundamental limits of performance \cite{Aguiar05a, Aguiar08}. Two approaches to path following have been dominantly  discussed in the literature: geometric control design methods and Lyapunov/backstepping techniques, see \cite{Banaszuk95a, Nielsen06a, Nielsen08a}, respectively, \cite{Dacic06, Dacic07, Dacic11a, Aguiar07a, Do04, Do06, Skjetne04, Skjetne05a}. The direct consideration of constraints on inputs and/or states, however, is difficult for either approaches.

To overcome this limitation nonlinear model predictive control  (NMPC) schemes tailored to path-following problems have been proposed. The early works \cite{ifat:faulwasser09a, ifat:faulwasser09c} as well as the results presented in \cite{Yu12a} are restricted to reference paths in the state space. 
This limits the applicability, since many realistic path-following problems---e.g. movement tasks for robots, autonomous vehicles, ship, and unmanned aerial vehicles---are defined in an output space rather than in the state space. Path following in output spaces is also termed \textit{output path following}.
Successful implementations of predictive output path following to real systems have been reported in \cite{Lam13a, epfl:faulwasser13b, Boeck13a}. Predictive output path following for underwater vehicles and non-holonomic systems is discussed in \cite{Alessandretti13} and \cite{Prodan13}.\footnote{
Besides predictive (feedback) control approaches to path-following problems, optimization-based feedforward path following has been in discussed in the literature  \cite{Shin85a, Verscheure09a, ifat:faulwasser11a, Lipp14a}.
These methods assume a special system structure---usually, it is required that the path is defined in a flat output space of a differentially flat system---and they are restricted to the computation of feedforward or open-loop controls.}

Besides practical considerations, the question of stability/convergence is challenging in output path following. First steps in this direction are presented in \cite{ifat:faulwasser10a, Lam10,Alessandretti13}. While in \cite{Lam10} recursive feasibility is lost due to contraction constraints, the preliminary results in \cite{ifat:faulwasser10a, Alessandretti13} draw upon terminal regions and end penalties to guarantee path convergence. However, a common drawback of these works is that no insight into the structure of constrained output path-following problems is provided.

\textcolor{black}{In the present contribution, we  extend and generalize previous results on predictive control for output path-following problems. In contrast to \cite{ifat:faulwasser09a, ifat:faulwasser09c,ifat:faulwasser10a,Yu12a, Lam10}, we investigate two different kinds of path-following problems, i.e., \textit{with and without} velocity assignments for the reference evolution. Similar to \cite{ifat:faulwasser10a,Alessandretti13}, we present a continuous-time sampled-data NMPC framework applicable to the design of controllers for output path-following problems under direct consideration of  constraints on states and inputs. Sufficient conditions based on terminal regions and end penalties guaranteeing the convergence to an output path as well as recursive feasibility of the arising optimization problems are provided. We extend our previous results \cite{ifat:faulwasser10a} by investigating the geometric nature of  path-following problems for nonlinear systems via the analysis of transverse normal forms and their use for the computation of stabilizing terminal regions and end penalties. This way, we provide a general framework for the design of continuous-time predictive control scheme for constrained output path-following problems.}

The remainder of the paper is structured as follows. In Section \ref{sec:problem} we outline the considered output path-following problems. Section \ref{sec:MPFC} contains the main contributions, i.e., a predictive control framework to path-following problems including sufficient convergence conditions. The design of suitable stabilizing terminal regions and end penalties is discussed in Section \ref{sec:Term}. To support our results we draw upon an example from robotics in Section \ref{sec:example}.

\subsection*{Notation}
The point-wise image of a set $\mcl{X} \subset \Rnx$ under a map $h:\Rnx \to \Rny$ is denoted as 
$h(\mcl{X}):=\{y \in \Rny | x \in \mcl{X} \mapsto y = h(x)\}$. 
The interior and the boundary of a set $\mcl{X}$ are denoted as $\inte{X}$, respectively, $\partial \mcl{X}$.
An open neighborhood of a point $x \in \Rnx$ is denoted as $\mcl{N}_x$.
The $k^{th}$ time derivative of a function $r: [t_0, \infty) \to  \mathbb{R}$ is written as $\frac{d^k r(t)}{dt^k}$ or more conveniently $r^{(k)}$.
$\mcl{C}^k$ denotes the set of $k$-times continuously differentiable functions.
The set of piece-wise continuous and right continuous functions on $\mbb{R}$ that take values in
$\mcl{V}\subset \mbb{R}^m$ is shortly denoted as $\mcl{P}\mcl{C}(\mcl{V})$.
The norm $\| x \|$  of $x \in \Rnx$ denotes the 2-norm. For $Q \in \mbb{R}^{n \times n}$  $\| x \|^2_Q = x^TQx$, while $\| Q\|$ denotes the induced 2-norm and $\|x\|_\infty$ denotes the infinity norm. 
The identity matrix of $\Rnx$ is written as ${I}^{n_x}$.  We use $\tilde{{I}}^{n_x} := \left(\begin{array}{c|c}{0}^{n_x-1,1} & {I}^{n_x-1} \\ \hline 0 & {0}^{1, n_x-1}\end{array}\right)$ and ${E}^{n_x} := (0, \dots, 0, 1)^T \in \Rnx$. $A = \diag(a_1, a_2, \dots, a_{n_x})$ denotes a diagonal matrix with entries $a_1, \dots, a_{n_x}$.

The solution  of an ordinary differential equation $\dot x = f(t,x,u)$, starting at time $t_0$ at $x(t_0) = x_0$
driven by an input $u: [t_0, \infty) \to \Rnu$, is written as $x(\cdot, t_0,
x_0|u(\cdot))$. The value of this solution at time $t_1\geq t_0$ is denoted as $x(t_1, t_0,
x_0|u(\cdot))$. The total derivative of a function $E(t, x(t))\in \mcl{C}^1$ with respect to $t$ is written as $\dfrac{d}{dt}\big[E(t, x(t))\big]:= \dfrac{\partial E}{\partial t} + \dfrac{\partial E}{\partial x}\dot x(t)$. The evaluation of $E\left(t, x(t)\right)$ at $t = t_1 + T$ is written as $E\left.\left(t,  x(t) \right)\right|_{t = t_1+T}$.

%% file: problem_statement.tex
\section{Path-following Problems} \label{sec:problem}
We consider nonlinear systems of the form
\begin{subequations} \label{eq:sys}
\begin{align} 
\dot x&=f(x) +\sum_{j=1}^{n_u}g_j(x)u_j, \quad x(t_0)=x_0  \label{eq:sys_dyn}\\
y & = h(x). \label{eq:sys_out}
\end{align}
\end{subequations}
 The map $h: \Rnx \to \Rny$ \eqref{eq:sys_out} defines  the output $y \in \mbb{R}^{n_y}$ or the variables of specific interest..
We assume that the maps $f: \Rnx \to \Rnx, g_j:  \Rnx \to \Rnx, h: \Rnx \to \Rny$ are sufficiently often continuously differentiable.
Here $x \in \mcl{X} \subseteq \Rnx$ and $u \in \mcl{U}\subset \Rnu$ denote the closed set of state constraints and the compact set of input constraints.  

Set-point stabilization usually refers to the task of stabilizing a fixed point in the state space. 
Trajectory tracking requires convergence of the states or the outputs of a system to a time-dependent reference that implies an explicit requirement \textit{when to be where} on the reference.
\footnote{Note that in the literature different terminologies are used for trajectory-tracking problems. For instance, if the task is to track a trajectory defined in an output space and the reference
trajectory is generated by an exogenous system (or exo-system), then one refers to the problem either
as \textit{model-following problem, servo problem} or as \textit{output regulation problem}, cf.  \cite{Anderson90, Isidori95a}.
We follow along the classic lines of \cite{Athans66} and deliberately denote all these cases as trajectory-tracking problems.}
 In contrast to trajectory-tracking problems we aim at driving the system along a geometric reference without any pre-specified timing information. 
This geometric reference is denoted as path $\mathcal{P}$. We assume it is given as a parametrized regular curve in the output space \eqref{eq:sys_out}
\begin{equation} \label{eq:path}
\mathcal{P} = \left\{y \in \Rny ~|~ \theta \in [\theta_0, \theta_1] \mapsto y = p(\theta)\right\}.
\end{equation}
Here the scalar variable $\theta$ is called the path parameter and $p:\mbb{R} \to \Rny$  is called a parametrization of $\mcl{P}$.
Note that the regularity of a geometric curve implies the local bijectivity of the parametrization $p(\theta)$, 
cf. \cite{Topogonov06}. 
The map $p:\mbb{R} \to \Rny$ is assumed to be sufficiently often continuously differentiable.  In general, the path parameter $\theta$ is time dependent but its time evolution $t \mapsto \theta(t)$ is not known a priori. 

Subsequently, path following refers to the problem of steering the output \eqref{eq:sys_out} to the path $\mcl{P}$ and to follow it along in direction
of increasing values of $\theta$. 
\textcolor{black}{
Obviously, one could solve this problem by choosing a fixed timing $\theta(t)$ and designing a trajectory-tracking controller for $p(\theta(t))$. This way, path following would be reformulated as a trajectory-tracking problem. However, the degree of freedom of adjusting $\theta(t)$ is lost. 
Here, we tackle the problem differently. The conceptual idea is to obtain the system input $u: [t_0, \infty) \to  \mcl{U}$ 
and the reference timing $t \mapsto \theta(t)$ in the controller, i.e., the controller determines the input $u(t)$ to converge to reference path as well as the time evolution $\theta(t)$ of the reference. }
In other words, we consider the following problem:
\vspace*{1mm}
\textit{
\begin{prob}[Constrained output path following] \label{prob:path} 
Given the system \eqref{eq:sys} and the reference path $\mcl{P}$ \eqref{eq:path}, design a controller  \textcolor{black}{that computes $u(t)$ and $\theta(t)$}  and achieves:
\begin{itemize}
	\item[i)] {Path convergence}: The system output $y =h(x)$ converges to the set $\mcl{P}$ \textcolor{black}{in the sense that}
	\[ \underset{t \to \infty}{\lim} \| h(x(t)) - p(\theta(t)) \| = 0.\]  		
	\item[ii)] {Convergence on path}: The system moves along $\mcl{P}$
	 in forward direction, i.e.
	\[ \dot\theta(t) \geq 0 \quad \textrm{and}\quad \underset{t \to \infty}{\lim} \| \theta(t) -\theta_1 \| = 0.\]  	
	\item[iii)] {Constraint satisfaction}:  The constraints on the states $x(t) \in \mcl{X}$ and the inputs
	 $u(t) \in \mcl{U}$ are satisfied for all times.
\end{itemize}
\end{prob}
}
Sometimes it might be desired to track a speed profile along the path. Following along the lines of \cite{Skjetne04, Aguiar05a, Aguiar08} such a problem is denoted as \textit{constrained output path following with velocity assignment}.
It differs from Problem \ref{prob:path} in part ii):
\vspace*{1mm}
\textit{
\begin{prob}[Output path following with velocity assignment] \label{prob:path_velo} 
Given the system \eqref{eq:sys} and the reference path $\mcl{P}$ \eqref{eq:path}, design a controller \textcolor{black}{that computes $u(t)$ and $\theta(t)$}, achieves part i) \& iii) of Problem \ref{prob:path} and guarantees:
\begin{itemize}
	\item[ii)] {Velocity convergence}: The path velocity $\dot\theta(t)$ converges to a predefined
	profile such that
	\[ \underset{t \to \infty}{\lim} \| \dot\theta(t) -\dot\theta_{ref}(t) \| = 0.\]  		
\end{itemize}
\end{prob}
}
\textcolor{black}{
Note that path following with velocity assignment is not equivalent to trajectory tracking, since path following with speed assignment does in general not specify a unique output reference $p(\theta(t))$. Rather it admits several reference trajectories  $p(\theta_i(t)),  i   \in\{1,2,\dots\}$, with $\dot\theta_i(t) =\dot\theta_{ref}(t)$,   which may differ with respect to $\theta$, i.e., $\theta_i(t) \neq \theta_j(t),\, i\neq j$.  }

A classical design of path-following controllers  regards the path parameter as a virtual state, whose evolution is determined through an additional ordinary differential equation (ODE) denoted as \textit{timing law}.  In essence, the timing law is an additional degree of freedom in the controller design. 
In backstepping approaches to path following, for instance, this timing law is constructed such that path convergence is enforced \cite{Do04, Skjetne04}. 
For sake of simplicity we use a simple integrator chain as timing law, i.e., the timing of the path parameter $\theta$ is specified via the ODE
\begin{equation} \label{eq:timing}
\theta^{(\hat r )}= v, \qquad \theta^{(i)}(t_0) = \theta^{(i)}_0, ~ i = 0, \dots, \hat r-1,
\end{equation} 
where, depending on the value of $\hat r$, the variable $v$ can be regarded as the speed, acceleration or jerk of the reference.
It is crucial to note that the time evolution $\theta(t)$---and thus also the evolution of the reference $p(\theta(t))$---can be controlled via the \textit{virtual} input $v: [t_0, \infty) \to \mcl{V}$.
At this point we do not specify the length $\hat r  \in \mbb{N}$ of the integrator chain \eqref{eq:timing}, which will depend on the design method and the system considered. We will come back to this issue in Section \ref{sec:Term}.

Relying on the timing law \eqref{eq:timing} we suggest to tackle path-following problems via the augmented system description

\begin{subequations}   \label{eq:sys_aug}
\begin{align}
\dot x & = f(x) +\sum_{j=1}^{n_u}g_j(x)u_j, \label{eq:sys_aug_x}\\
\dot z & = \tilde{{I}}^{\hat r}z +{E}^{\hat r}v 
\label{eq:sys_aug_z}\\
e & = h(x) - p(z_1), \label{eq:sys_aug_e}\\
 \theta & = z_1. \label{eq:sys_aug_theta}
\end{align}
\end{subequations}
Here, \eqref{eq:sys_aug_x} includes the dynamics of the system to be controlled \eqref{eq:sys}, \eqref{eq:sys_aug_z} is the timing law \eqref{eq:timing}  with $z = (\theta, \dot \theta, \dots, \theta^{(\hat r-1)})^T$.
The error output  \eqref{eq:sys_aug_e} represents the deviation from the path, while \eqref{eq:sys_aug_theta} describes the current reference position on the path.

%% file: mpfc.tex
\section{Model Predictive Path-following Control} \label{sec:MPFC}
Subsequently, we propose a predictive path-following control scheme to tackle path-following problems.  We denote this scheme as model predictive path-following control (MPFC).
We will first focus the investigations on Problem \ref{prob:path}, including the presentation of sufficient convergence conditions. The extension to 
path following with speed assignment (Problem \ref{prob:path_velo}) is discussed at the end of this section.

\subsection{Proposed Predictive Control Scheme}
As standard in predictive control the applied input is based on repeatedly solving an optimal control problem (OCP). That is, at each sampling instance $t_k = k\delta, ~ k \in \mbb{N}_0, \delta >0$ we solve an OCP that minimizes the cost functional
\begin{multline} \label{eq:OCP_J}
 J\left(x(t_k), \bar z(t_k),   \bar u_k(\cdot), \bar v_k(\cdot)\right) \\
    = \int_{t_k}^{t_k +T}
      F\left(\bar e(\tau), \bar \theta(\tau), \bar u_k(\tau),
			\bar v_k(\tau)  \right)d\tau \\
		+ E\left.\left(t, \bar x(t), \bar z(t) \right)\right|_{t = t_k+T}.
\end{multline}
As usual in NMPC the function $F: \Rny\times\mbb{R}\times\mcl{V}\times\mcl{U} \to \mbb{R}_0^+$ is termed cost function, and $E:  \mbb{R}^+_0\times\Rnx\times\Rnz
 \to \mbb{R}_0^+$ is denoted as terminal or end penalty; 
predicted states and inputs are indicated by the superscript $\bar{\cdot}$. The subscript $\cdot_k$ indicates that an open-loop input $\bar u_k(\cdot)$ is computed at the $k^{th}$ sampling instant $t_k$. The constant $T \in(\delta, \infty)$ is called the prediction horizon. 
The OCP to be solved in a receding horizon fashion at the sampling times $t_k$ reads:
\begin{subequations} \label{eq:OCP}
\begin{equation} \label{eq:OCP_min}
 \underset{(\bar u_k(\cdot), \bar v_k(\cdot))\in \mcl{P}\mcl{C}(\mcl{U}\times\mcl{V})}{\minimize} ~ J\left(x(t_k), \bar z(t_k), \bar u_k(\cdot), \bar v_k(\cdot)\right) 
\end{equation}
subject to $ \forall \tau \in [t_k, t_k + T]:$
\begin{align}
     \dot{\bar{x}}(\tau)& = f(\bar x(\tau)) +\sum_{j=1}^{n_u}g_j(\bar x(\tau))\bar u_{k,j}(\tau)), \quad \bar x(t_k) =
x(t_k)\phantom{,z(t_{k-1}}  \label{eq:OCP_x}\\
 \quad \dot{\bar{z}}(\tau)& = \tilde{{I}}^{\hat r}z(\tau) + {E}^{\hat r}v_k(\tau) \label{eq:OCP_z} \\
 \bar{z}(t_k) &= \bar z(t_k, t_{k-1},\bar z(t_{k-1})| \bar v^\star_{k-1}(\cdot)) \label{eq:OCP_z0}\\
 \quad \bar e(\tau) &= h(\bar x(\tau)) - p(\bar z_1(\tau)) \label{eq:OCP_e}\\
 \quad \bar \theta(\tau)& = \bar z_1(\tau) \label{eq:OCP_theta} \\
 \quad \bar x(\tau) &\in \mcl{X}, ~\bar 	u_k(\tau)  \in \mcl{U}  \label{eq:OCP_x_con} \\
 \quad \bar z(\tau) &\in \mcl{Z}, ~\bar 	v_k(\tau)  \in \mcl{V} \label{eq:OCP_z_con} \\
	(\bar x(t_k &+ T), \bar z(t_k + T))^T \in \mcl{E}\subset \mcl{X}\times \mcl{Z}.\label{eq:OCP_term_con}
\end{align}
\end{subequations}
For sake of simplicity we assume that an optimal solution to  OCP \eqref{eq:OCP} exists and is attained. The statement of conditions, which ensure the existence of optimal solutions is beyond the scope of this paper, instead we refer to\cite{Lee67, Berkovitz74}.
Note that the decision variables of the minimization in \eqref{eq:OCP_min} are the real system input $u(\cdot)  \in \mcl{P}\mcl{C}(\mcl{U})$ as well as the virtual path parameter input $v(\cdot)  \in \mcl{P}\mcl{C}(\mcl{V})$. 
In other words, by solving \eqref{eq:OCP} we obtain the system input and the reference evolution at the same time.
State and input constraints of the system to be controlled are enforced by \eqref{eq:OCP_x_con}. Furthermore, the path parameter dynamics \eqref{eq:OCP_z} are subject to the state and input constraints \eqref{eq:OCP_z_con}, whereby the state constraint $\mcl{Z}$
is defined as
\begin{equation}   \label{eq:Z} 
 \mcl{Z} :=  [\theta_0, \theta_1] \times \mbb{R}_0^+ \times \mbb{R}^{\hat
r -2} \subset \mbb{R}^{\hat r}.
\end{equation}
Essentially, this constraint ensures that $\bar\theta = \bar z_1 \in [\theta_0, \theta_1]$,
as well as $\dot{\bar \theta} \geq 0$. This way we enforce monotonous forward motion along the path. In order to avoid impulsive solutions of the
path parameter dynamics  \eqref{eq:OCP_z} the admissible values of the virtual path parameter
inputs $\bar v$ are restricted  to a compact set $\mcl{V} \subset \mbb{R}$
containing $0$ in its interior in \eqref{eq:OCP_z_con}.

While at each sampling instance the measured state information
$x(t_k)$ serves as initial condition for \eqref{eq:OCP_x},
the initial condition of the timing law \eqref{eq:OCP_z} is based on the last predicted trajectory
$\bar{z}(\cdot, t_{k-1},\bar{z}(t_{k-1})|\bar v^\star_{k-1}(\cdot))$ evaluated at time $t_k$. 
In cases where no initial condition for the first sampling instance $k=0$ is given, 
we obtain $\bar z(t_0)$ via 
\begin{subequations}
\begin{align} \label{eq:z0}
\bar z(t_0) &= (\theta(t_0), 0, \dots, 0)^T \\
\theta(t_0) &=  \argmin{\theta \in [\theta_0, \theta_1]} ~\| h(x_0) -
p(\theta)\|. 
\end{align}
\end{subequations}
In general, this problem might have multiple optimal solutions, and we simply choose one of them. 

Similar to classical NMPC schemes \cite{Chen98, Fontes01, Mayne00a} the terminal constraint \eqref{eq:OCP_term_con} enforces  that at the end of each optimization the predicted augmented state $(\bar x(t_k+T), \bar z(t_k+T))^T$ 
lies inside a terminal region $\mcl{E} \subseteq \mcl{X}\times\mcl{Z}$. Although only outputs and inputs are penalized in the cost function $F$ in \eqref{eq:OCP_J}, the terminal constraint is stated in the state space. 
The reason for this choice is that---under suitable assumptions---output path following can be reformulated as a manifold stabilization problem in the state space, cf. \cite{Nielsen08a, ifat:faulwasser13a_short}. We will investigate this issue in detail in Section \ref{sec:Term}.
 Also note that the terminal penalty $E$  will be used to obtain an upper bound on the cost associated to solutions originating inside the terminal region $\mcl{E} \subseteq \mcl{X}\times\mcl{Z}$. Thus $E$ is stated as a function of the augmented state  $(x,z)^T$.  Additionally, and without loss of generality, we consider explicit time dependence of $E$ in \eqref{eq:OCP_J}. 

The optimal solution of \eqref{eq:OCP} is denoted as
$J^{\star}\left(x(t_k), \bar z(t_k),  \bar u^\star_k(\cdot), \bar v^\star_k(\cdot)\right)$. It is specified by optimal input trajectories $\bar u^{\star}_k: [t_k, t_k+T] \to \mcl{U}$ and $\bar
v^{\star}_k: [t_k, t_k+T] \to \mcl{V}$. 
\textcolor{black}{
Now, we are ready to summarize the MPFC scheme in Figure \ref{fig:MPFC}. As usual in NMPC, in Step 1 we need to obtain (observed or measured) state information. In Step 2 we solve the OCP \eqref{eq:OCP}. }
\begin{figure}[t]
\begin{center}
\begin{algorithmic}
\STATE \textbf{Data}: $x(t_0), \bar z(t_0), \delta, T$ \vspace*{0.25cm}
\STATE \textbf{Step 0:} Initialize $k = 0$. 
\STATE \textbf{Step 1:} Get state information $x(t_k)$.
\STATE \textbf{Step 2:} Solve OCP \eqref{eq:OCP} with initial condition $x(t_k), \bar z(t_k)$. 
\STATE \textbf{Step 3:}  Apply optimal input \[
\forall t \in [t_k, t_k + \delta): \qquad u(t) = \bar u^{\star}_k(t).
\]
\STATE \textbf{Step 4:} Assign $\bar{z}(t_{k+1}) = \bar z(t_{k+1}, t_{k},\bar z(t_{k})| \bar v^\star_{k}(\cdot))$. 
\STATE \textbf{Step 5:} $ k \to k+1$ \textbf{Goto} Step 1. 
\end{algorithmic}
\end{center}
\caption{MPFC scheme based on OCP \eqref{eq:OCP}. \label{fig:MPFC}}
\end{figure}
\textcolor{black}{
And in Step 3, the first part of the optimal input $\bar u^{\star}_k(\cdot)$ is applied to the real system \eqref{eq:sys}  until the next sampling time. Note that the virtual input $\bar v^{\star}_k(\cdot)$ and the path parameter state $\bar z(\cdot)$ are merely internal controller variables, i.e., in Step 4 the next iteration is prepared.
\begin{rema}[Dynamic nature of the MPFC scheme]
It should be noted that the solution to \eqref{eq:OCP} at time $t_k$ depends on the solution at the previous sampling instant $t_{k-1}$. The reason is that the initial condition of $\bar z$ at time $t_k, k >0$ is based on the last predicted trajectory $\bar{z}(\cdot, t_{k-1},\bar{z}(t_{k-1})~|~\bar v_k^\star(\cdot))$ evaluated at time $t_k$, cf. \eqref{eq:OCP_z0} and Step 4 shown in Figure \ref{fig:MPFC}. In other words, the path parameter state $\bar z$ is as an internal state of the MPFC scheme. Thus, in contrast to usual NMPC schemes for set-point stabilization such as \cite{Chen98, Fontes01, Mayne00a}, the MPFC scheme as depicted in Figure \ref{fig:MPFC} is a dynamic feedback strategy.
\end{rema}
}
\textcolor{black}{
\begin{rema}[Computational demand]
We remark that the present paper is focused on the concept of predictive path following and its properties. Thus, the efficient numerical implementation of the proposed scheme is beyond its scope. However, note that \eqref{eq:OCP} is a typical OCP for an NMPC scheme with terminal constraints and terminal penalties. The only difference  compared to NMPC for set-point stabilization are the increased state and input dimensions. Thus, to solve  \eqref{eq:OCP} one may apply existing numerical tools tailored for real-time feasible NMPC with state and terminal constraints, cf. also the successful implementations of NMPC for path following in \cite{Lam13a, epfl:faulwasser13b, Boeck13a}. 
\end{rema}
}

\subsection{Sufficient Convergence Conditions} 
As is well known the receding horizon application of optimal open-loop inputs does not necessarily lead to stability nor to convergence of the closed-loop output to the path \cite{Mayne00a, Fontes01}. Thus we are interested in conditions ensuring that the MPFC scheme \eqref{eq:OCP} solves Problem \ref{prob:path}.
In order to present such conditions we rely on the following assumptions.
\textit{
\begin{ass}[System dynamics] \label{ass:sys}
 The vector fields $f:\Rnx \to\Rnx$ and $g_j: \Rnx \to \Rnx, ~ j = 1, \dots, n_u$ from \eqref{eq:sys}  are continuous and locally Lipschitz
for any pair $(x,u)^T \in \mcl{X} \times \mcl{U}$. 
\end{ass}
\vspace*{0.2cm}
\begin{ass}[Continuity of system trajectories] \label{ass:cont_sol} 
 For any $x_0 \in \mcl{X}$  and any input
function $u(\cdot)\in \mcl{P}\mcl{C}(\mcl{U})$ the system \eqref{eq:sys} has
an absolutely continuous solution.
\end{ass}
\vspace*{0.2cm}
\begin{ass}[Consistency of path and state constraints] \label{ass:p_consist} 
The path $\mcl{P}$ from \eqref{eq:path} is contained in the interior of the
point-wise image of the state constraints $\mcl{X}$ under the output map $h: \Rnx
\to \Rny$ from \eqref{eq:sys_out}, i.e.,  
$\mcl{P} \subset \interior(h(\mcl{X}))$.
\end{ass}
\vspace*{0.2cm}
\begin{ass}[Cost function] \label{ass:F} 
 The cost function $F: \Rny\times\mbb{R}\times\mcl{V}\times\mcl{U} \to \mbb{R}_0^+$ is continuous. Furthermore, we assume that $F$ is lower bounded by a class $\mcl{K}$ function, i.e., $\underline{\psi}(\|e, \theta-\theta_1 \|) \leq F(e,\theta,u,v) $.\footnote{In essence one could write $F$  more general as a function of $x,z,u$ and $v$. Here, we focus explicitly on cost functions depending on $e$ and $\theta$ to highlight that we do not consider a set-point stabilization problem but a more general (path-following) problem whereby merely outputs are penalized in the cost function.}
\end{ass}
}
\vspace*{0.2cm}
Assumptions \ref{ass:sys}-\ref{ass:cont_sol} are very similar to the ones made for NMPC for set-point stabilization problems, cf. \cite{Fontes01, Chen98}.  Basically, these assumptions are used to guarantee the local existence and uniqueness of solutions of \eqref{eq:sys}. \textcolor{black}{Assumption \ref{ass:cont_sol} is made in order to apply Barbalat's Lemma in a crucial step of the proof of Theorem \ref{thm:oMPFC}.}
 Assumptions  \ref{ass:p_consist}--\ref{ass:F} are specific for model predictive path-following control. The former  
is necessary to avoid cases for which parts of the path are inconsistent with the state constraints. 
The latter assumption 
requires that the cost function is lower bounded in terms of the path-following error and the path parameter. This way, we enforce path convergence as well as convergence on the path.
Under the above assumptions the following result is obtained. 
\textit{
\begin{thm}[Convergence of MPFC] \label{thm:oMPFC} 
Consider Problem \ref{prob:path} and suppose that Assumptions 
\ref{ass:sys}--\ref{ass:F} hold.
Suppose that a terminal region
		$\mcl{E} \subset \mcl{X}\times\mcl{Z}$ and a terminal penalty 
		$E(t, x, z)$ exist such that the following conditions are satisfied:
\begin{itemize}
\item [i)]	The set $\mcl{E}$ is compact. $E(t, x, z)$  is  $\mcl{C}^1$ and positive semi-definite  with respect to $(t,x,u)$.					 
\item[ii)]	For all $t \in [t_0, \infty)$ and all $(\tilde x, \tilde z)^T \in \mcl{E}$ there exists a scalar $\epsilon \geq \delta >0$
 and admissible inputs $(u_{\mcl{E}}(\cdot),v_{\mcl{E}}(\cdot)) \in \mcl{P}\mcl{C}(\mcl{U}\times\mcl{V})$ 
		such that for all $\tau \in [t, t+\delta]$
\begin{equation}
	\dfrac{d}{d\tau}\big[E(\tau, x(\tau),z(\tau))\big]
		   + F\big( e(\tau), \theta(\tau), u_{\mcl{E}}(\tau),v_{\mcl{E}}(\tau)\big) \leq  0,
		   \label{eq:ass2}
\end{equation}
	    and the solutions $x(\tau) = x(\tau,t,\tilde x|u_{\mcl{E}}(\cdot))$ and 
	    $z(\tau)= z(\tau,t,\tilde z|v_{\mcl{E}}(\cdot))$, starting at 
	    $(\tilde x, \tilde z)^T \in \mcl{E}$, stay in $\mcl{E}$ for all $\tau \in [t, t+\delta]$.
\item[iii)] The OCP \eqref{eq:OCP} is feasible for $t_0$.	
	\end{itemize}
	Then the MPFC scheme depicted in Figure \ref{fig:MPFC} solves Problem \ref{prob:path}. 
\end{thm}
}
\begin{proof}
In essence the proof of this result can be obtained via a reformulation of the standard results on convergence of continuous time NMPC for set-point stabilization, see e.g. \cite{Fontes01, Chen98, Mayne00a}. Thus we provide only a shortened proof here outlining the main differences to  \cite{Fontes01}.

{Step 1} (Recursive feasibility):  In the first step recursive feasibility is shown via the usual concatenation of optimal inputs $(u^\star_k(\cdot), v^\star_k(\cdot))$ with the terminal controls $(u_{\mcl{E}}(\cdot),v_{\mcl{E}}(\cdot))$ as in \cite{Fontes01}. Since these concatenated inputs ensure positive invariance of the terminal constraint $\mcl{E}$ it immediately follows that the MPFC scheme based on \eqref{eq:OCP} is recursively feasible. 

{Step 2} (Constraint satisfaction and forward motion):  In the second step we verify that ii)-iii) of Problem \ref{prob:path} are satisfied. Recall that the terminal constraint set is contained in the state constraints, i.e.,  $\mcl{E} \subset \mcl{X}\times\mcl{Z}$. Thus part iii) of Problem \ref{prob:path} is satisfied. Furthermore, for all $(x,z)^T\in\mcl{E}$ we have $z \in \mcl{Z}$ from \eqref{eq:Z}, which implies that also the forward motion requirement $\dot \theta = z_2 \geq 0$ holds.  Hence part ii) of Problem \ref{prob:path} is ensured. 

{Step 3} (Path convergence and convergence on path): It remains to verify that path convergence (part i) of Problem \ref{prob:path}) is guaranteed. This is done in the third step. 
First, we consider the value function of OCP \eqref{eq:OCP}
\[
V(t_k, x(t_k), \bar z(t_k)) := J\left(x(t_k),  \bar z(t_k), u_k^\star(\cdot),  v_k^\star(\cdot)\right).
\]
Similar to \cite[Lemma 5]{Fontes01} one uses the invariance condition \eqref{eq:ass2} to show that for all sampling times $\delta \in (0, \epsilon]$ we have
\begin{multline*}
V(t_{k+1}, x(t_{k+1}), \bar z(t_{k+1})) - V(t_k, x(t_k), \bar z(t_k)) \\ \leq - \int_{t_k}^{t_{k+1}} \underline{\psi}(\|e(t), \theta(t)-\theta_1\|)dt.
\end{multline*}

Second, we consider the MPC value function
\begin{multline*} \label{eq:def_Vd}
V^\delta(t, x(t), \bar z(t)) := \\ V(t_k, x(t_k), \bar z(t_k)) -\int_{t_k}^t F( e(\tau), \theta(\tau), u_k^\star(\tau),  v_k^\star(\tau)) d\tau,
\end{multline*}
 which is the remainder of $V(t_k, x(t_k), \bar z(t_k))$  for $x(t) = x(t, t_k, x(t_k) |u^\star_k(\cdot))$ and $\bar z(t) = \bar z(t, t_k, \bar z(t_k) |v^\star_k(\cdot))$. In
the definition of $V^\delta(t, x(t),\bar z(t))$ the time instant is $t_k = k\delta$ with 
 $k =  \underset{k\in \mbb{N}}{\max}\{k ~|~ t_k \leq t\}$, i.e., the closest previous sampling instant.
One can apply the same ratio as in \cite[Lemma 6]{Fontes01} to show that for all $t \geq t_0$ it holds that
\begin{multline*}
 V^\delta(t, x(t), \bar z(t)) + \int_{t_0}^{t} \underline{\psi}(\| e(\tau), \theta(\tau)-\theta_1\|) d\tau \\ \leq  V^\delta(t_0, x(t_0), \bar z(t_0)).
\end{multline*}
Finally, we use Assumption \ref{ass:cont_sol} and apply Barbalat's Lemma \cite[Lemma 4]{Michalska94} to establish convergence $\lim\limits_{t \to \infty}\|e(t), \theta(t)-\theta_1\| =0$.
This finishes the proof. 
\end{proof}
\vspace*{0.2cm}
Note that the proposed control scheme aims on convergence
of the output $y=h(x)$ to the path and not on \textit{Lyapunov}-like state
stability.\footnote{Even for sampled-data continuous-time NMPC tailored to set-point stabilization it is in general difficult to prove Lyapunov stability. Usually, merely asymptotic convergence is established  \cite{Fontes01}. This is due to the fact that between two sampling instances $t_k$ and $t_{k+1}$ the controller applies open-loop inputs to the system.}  In other words, Theorem \ref{thm:oMPFC} allows
cases where the output converges to the path while the states might move through
$\mcl{X}\times\mcl{Z}$.  This  means that general cases, in which the internal dynamics of  \eqref{eq:sys_aug} 
with respect to the output $(e, \theta)^T$ are merely bounded in $\mcl{X}\times\mcl{Z}$ but not asymptotically convergent, are possible. At the end of each finite prediction horizon, however, the predicted states have to reach the terminal constraint $\mcl{E} \subset \mcl{X}\times\mcl{Z}$. This 
 implies that in the nominal case without plant-model mismatch all states of \eqref{eq:sys_aug}---which includes the states of the zero dynamics of \eqref{eq:sys_aug} with respect to the output $(e,\theta)^T$---are bounded. Thus the fact that we merely penalize outputs in the cost function $F$ does not lead to further difficulties.
It is also straightforward to see that Theorem \ref{thm:oMPFC} holds for the special, and usually hardly application relevant, case of invertible output maps $h: \Rnx \to \Rnx$ and for paths directly defined in the state space.

\begin{rema}[Non-input-affine systems] \label{rem:fxu}
We point out that the proof of Theorem \ref{thm:oMPFC} does not rely on the specific input-affine structure of the system \eqref{eq:sys}. 
Indeed, the conditions of the theorem hold even for cases of non-input-affine systems, i.e., general systems of the form $\dot x = f(x,u), ~y =h(x)$.
The main reason to consider input-affine systems is that this choice allows further insight into the geometric nature of path-following problems. And, as we will show subsequently, this restriction of the considered system class simplifies the computation of end penalties and terminal regions satisfying the conditions of Theorem \ref{thm:oMPFC}.
\end{rema}

\subsection{Extension to Predictive Path Following with Velocity Assignment}
At this point it is fair to ask how the result of Theorem \ref{thm:oMPFC} can be extended to velocity-assigned path following as described in Problem \ref{prob:path_velo}.
To this end we modify Assumption \ref{ass:F} as follows:
\textit{
\begin{ass}[Cost function] \label{ass:F_velo} 
 The cost function $F: \Rny\times\mbb{R}\times\mcl{V}\times\mcl{U} \to \mbb{R}_0^+$ is continuous. Furthermore, we assume that $F$ is lower bounded by a class $\mcl{K}$ function, i.e., $\underline{\psi}(\|e, \dot\theta-\dot\theta_{ref} \|) \leq F(e,\dot\theta,u,v) $.
\end{ass}
}
For velocity-assigned path-following problems the path parameter $\theta = z_1$ might grow unbounded. Thus the constraint $z \in \mcl{Z}$, cf. \eqref{eq:OCP_z_con},  on the path parameter states should be dropped. 
The next result states that a modified MPFC scheme, in which a cost function according to Assumption \ref{ass:F_velo} and no path parameter state constraint ($\mcl{Z} = \mbb{R}^{\hat r}$) are considered, solves Problem \ref{prob:path_velo}.
\textit{
\begin{thm}[Convergence of MPFC with velocity assignment] \label{thm:oMPFC_velo} 
Consider Problem \ref{prob:path_velo} and suppose that Assumptions 
\ref{ass:sys}--\ref{ass:p_consist} and \ref{ass:F_velo} hold.
Suppose that $\mcl{Z} = \mbb{R}^{\hat r}$,  a terminal region
		$\mcl{E} \subset \mcl{X}\times \mbb{R}^{\hat r}$ and a terminal penalty 
		$E(t, x, z)$ exist such that the following conditions are satisfied:
\begin{itemize}
\item [i)]	The set $\mcl{E}$ is compact with respect to $x$ and closed with respect to $z$. $E(t, x, z)$  is  $\mcl{C}^1$ with respect to $(t,x,u)$ and positive semi-definite.					 
\item[ii)]	For all $t \in [t_0, \infty)$ and all $(\tilde x, \tilde z)^T \in \mcl{E}$ there exist a scalar $\epsilon \geq \delta >0$
 and admissible inputs $(u_{\mcl{E}}(\cdot),v_{\mcl{E}}(\cdot)) \in \mcl{P}\mcl{C}(\mcl{U}\times\mcl{V})$ 
		such that for all $\tau \in [t, t+\delta]$
\begin{equation}
	\dfrac{d}{d\tau}\big[E(\tau, x(\tau),z(\tau))\big]
		   + F\big( e(\tau), \dot\theta(\tau), u_{\mcl{E}}(\tau),v_{\mcl{E}}(\tau)\big) \leq  0,
		   \label{eq:ass3}
\end{equation}
	    and the solutions $x(\tau) = x(\tau,t, \tilde x|u_{\mcl{E}}(\cdot))$ and 
	    $z(\tau)= z(\tau,t,\tilde z|v_{\mcl{E}}(\cdot))$, starting at 
	    $(\tilde x, \tilde z)^T \in \mcl{E}$, stay in $\mcl{E}$ for all $\tau \in [t, t+\delta]$.
\item[iii)] The OCP \eqref{eq:OCP} is feasible for $t_0$.	
	\end{itemize}
	Then the MPFC scheme depicted in Figure \ref{fig:MPFC} solves Problem \ref{prob:path_velo}. 
\end{thm}
}
The proof of this result is similar to the proof of Theorem \ref{thm:oMPFC}. Recursive feasibility and constraint satisfaction can be shown along the same lines. The only difference is that in the last step the application of Barbalat's Lemma leads to the conclusion that $\lim\limits_{t \to \infty}\|e(t), \dot\theta(t)-\dot\theta_{ref}(t)\| =0$.

%% file: term.tex
\section{Design of Suitable Terminal Regions \\and End Penalties} \label{sec:Term}
So far we have shown that a suitable combination of a terminal region and an end penalty can be used to guarantee convergence of predictive path following. 
Similar to the case of NMPC for stabilization and tracking problems, \cite{Chen98, Michalska93,ifat:faulwasser11b}, the design of terminal regions and corresponding end penalties is challenging for the proposed MPFC scheme. In general, the computation of terminal regions involves the design of a locally admissible controller. 
Subsequently, we present two technical results that allow using trivial end penalties $E(t,x(t),z(t)) = 0$. And later we discuss the inherent geometric properties of path-following problems. 

\subsection{Trivial End Penalties}
As a preparation step we introduce the notation 
\begin{equation*}
\varphi_\mcl{E}(t)  
:= F( e(t),  \theta(t), u_{\mcl{E}}(t), v_{\mcl{E}}(t))
\end{equation*}
describing the evolution of the cost function for given inputs $(u_{\mcl{E}}(\cdot), v_{\mcl{E}}(\cdot)) \in \mcl{P}\mcl{C}(\mcl{U}\times\mcl{V})$.
\textit{
\begin{lemm}[Existence of a time-dependent terminal penalty] \label{lem:ExistE} 
\textit{
Assume that there exist terminal controls $(u_{\mcl{E}}(\cdot), v_{\mcl{E}}(\cdot)) \in \mcl{P}\mcl{C}(\mcl{U}\times\mcl{V})$ defined for all $t \in [t_0, \infty)$ and a compact terminal region $\mcl{E} \subset \mcl{X} \times {Z}$ such that the following conditions hold:
\begin{itemize}
\item[i)] The set $\mcl{E} \subset \mcl{X} \times {Z}$ is rendered controlled positively invariant by $(u_{\mcl{E}}(\cdot), v_{\mcl{E}}(\cdot))$, i.e., any solution  $x(t) = x(t, t_0, \tilde x\,|\,u_{\mcl{E}}(\cdot))$ and 
	    $z(t)= z(t, t_0, \tilde z \,|\,v_{\mcl{E}}(\cdot))$, starting at time $t_0$ at any $(\tilde x, \tilde z)^T \in \mcl{E}$, stays in
	     $\mcl{E}$ for all $t \in [t_0, \infty)$.
\item[ii)] For all $(\tilde x, \tilde z)^T \in \mcl{E}$ and all $t \geq t_0$ it holds that
\begin{equation}
\varphi_\mcl{E}(t) \leq c(\tilde x, \tilde z) e^{-\alpha(\tilde x, \tilde z) (t-t_0)}
\end{equation}
and  for all $(\tilde x, \tilde z)^T \in \mcl{E}:$
\[ 0<\underline \alpha \leq\alpha(\tilde x, \tilde z) \textrm{ and } 0 \leq c(\tilde x, \tilde z)  \leq \bar c < \infty.\]
\end{itemize}
Then, 
there exists an end penalty $\tilde E: [t_0,\,\infty) \to \mbb{R}^+\backslash\infty$, such that $\tilde E$ and $\mcl{E}$ satisfy the conditions of Theorem \ref{thm:oMPFC}.
}
\end{lemm}
}
\begin{proof}
Without difficulties if follows from part ii) of the lemma that for all $(\tilde x, \tilde z)^T \in \mcl{E}$ and all $t \in [t_0, \infty)$
\[
\varphi_\mcl{E}(t)
\leq c(\tilde x, \tilde z) e^{-\alpha(\tilde x, \tilde z)(t-t_0)} \leq c(\tilde x, \tilde z) e^{-\underline{\alpha}(t-t_0)} \leq \overline{c} e^{-\underline{\alpha}(t-t_0)}.
\]
It is easy to verify that for all $(\tilde x, \tilde z)^T \in \mcl{E}$
\begin{equation} \label{eq:tildeE}
\tilde{E}(t) = \overline{c}\underline{\alpha}^{-1} e^{-\underline \alpha (t-t_0)}
\end{equation}
satisfies the cost decrease condition \eqref{eq:ass2}.
\end{proof}
Now, consider two variants of the objective functional \eqref{eq:OCP_J} differing only by the end penalty
\begin{align*}
J_1\left( \bar u_k(\cdot), \bar v_k(\cdot)\right) &= \int_{t_k}^{t_k+ T}
      F\left(\bar e(\tau), \bar \theta(\tau), \bar u_k(\tau),\bar v_k(\tau)  \right)d\tau, \\
 J_2\left(\bar u_k(\cdot), \bar v_k(\cdot)\right) &= \int_{t_k}^{t_k +T} 
      F\left(\bar e(\tau), \bar \theta(\tau), \bar u_k(\tau),\bar v_k(\tau)  \right)d\tau \\&\phantom{=}+\tilde E(t_k+T).
\end{align*}
Note that for sake of simplified notation we neglect the dependence of $J_{i}, i \in\{1,2\}$ on $x(t_k), \bar z(t_k)$. We denote two variants of OCP \eqref{eq:OCP} as follows: \eqref{eq:OCP} with $J_1$ is denoted as OCP$_1$ and \eqref{eq:OCP} with $J_2$ is denoted as OCP$_2$.

\begin{lemm}[Equivalence of optimal solutions] \label{lem:OCP_equivalence}\textit{\quad Consider OCP$_i,$ $i \in \{1,2\}$ subject to the same initial condition $x(t_k), \bar z(t_k)$. The following two statements hold:
\begin{itemize}
\item[i)] Suppose that OCP$_1$ has an optimal solution $\bar u_k^\star(\cdot), \bar v_k^\star(\cdot)$, then $\bar u_k^\star(\cdot), \bar v_k^\star(\cdot)$ is also an optimal solution to 
OCP$_2$.
\item[ii)] Suppose that OCP$_2$ has an optimal solution $\bar u_k^\star(\cdot), \bar v_k^\star(\cdot)$, then $\bar u_k^\star(\cdot), \bar v_k^\star(\cdot)$ is also an optimal solution to 
OCP$_1$.
\end{itemize}
}
\end{lemm}
\begin{proof}
We first consider statement i).
For any admissible choice of $\bar u_k(\cdot), \bar v_k(\cdot)$ it holds that
\[
 J_2\left(\bar u_k(\cdot), \bar v_k(\cdot)\right) - J_1\left( \bar u_k(\cdot), \bar v_k(\cdot)\right) = \tilde E(t_k+T).
\]
This means that, for any sampling instant $t_k$ and any initial condition $x(t_k), \bar z(t_k)$, the two objective functionals $J_{i}, i \in\{1,2\}$ only differ by a constant. And since $\tilde{E}(t)$ from \eqref{eq:tildeE} depends only on $t$ and not on $x$ or $z$, the value of $\tilde E(t_k+T)$ is not influenced by the choice of $\bar u_k(\cdot), \bar v_k(\cdot)$. 
Hence, any input $\bar u_k^\star(\cdot), \bar v_k^\star(\cdot)$, which is an optimal solution to OCP$_1$, is also an optimal solution to OCP$_2$.  The proof of statement ii) is obtained without difficulties based on similar arguments.
\end{proof}
The next result shows how the last two lemmas can be combined.  

\begin{prop}[Trivial end penalty $E(t, x(t), z(t)) = 0$] \label{prop:E(t)}
\textit{ \\
Suppose that there exist terminal controls $(u_{\mcl{E}}(\cdot), v_{\mcl{E}}(\cdot)) \in \mcl{P}\mcl{C}(\mcl{U}\times\mcl{V})$ defined for all $t \in [t_0, \infty)$ and a compact terminal region $\mcl{E} \subset \mcl{X} \times {Z}$ such that conditions i)--ii) of Lemma \ref{lem:ExistE} hold. 
}
\textit{
Then the MPFC scheme using the end penalty $E(t,x(t),z(t)) = 0$ and the terminal region $\mcl{E}$ solves Problem \ref{prob:path}. 
}
\end{prop}
\begin{proof}
From Lemma \ref{lem:ExistE} we know that  $\tilde E$ and $\mcl{E}$ satisfy the conditions of Theorem \ref{thm:oMPFC}, i.e., they enforce path convergence and convergence on the path. From Lemma \ref{lem:OCP_equivalence} we know that using $E(t) = 0$ instead of $\tilde E(t)$ in OCP \eqref{eq:OCP} we obtain inputs that are optimal for OCP \eqref{eq:OCP} with $\tilde E(t)$. Thus, applying the end penalty $E(t,x(t),z(t)) = 0$ combined with the terminal region $\mcl{E}$, we obtain the conclusions of Theorem \ref{thm:oMPFC}.
\end{proof}

\begin{rema}[Exponentially stabilizing terminal controls laws] ~\\
The main insight obtained by the last proposition can be summarized as follows: from the stability point of view, terminal controls, which ensure exponential cost decrease in the sense of Lemma \ref{lem:ExistE}, render it unnecessary to determine terminal penalties. However, suitably chosen terminal penalties, can improve closed-loop performance. 
Thus, it is not surprising that the last proposition can be adjusted to other NMPC schemes designed for stabilization or trajectory tracking, cf. \cite{ifat:faulwasser11b}. 
\end{rema}
\subsection{Geometric Structure of Path-following Problems}
Next, we show that  path-following problems are equivalent to the problem of stabilizing a certain manifold in the state space, see also \cite{Skjetne04, Nielsen08a, ifat:faulwasser13a_short}.  For sake of simplicity the considerations focus on constrained path following without velocity assignment (Problem \ref{prob:path}). Corresponding results can be also established for path following with velocity assignment (Problem \ref{prob:path_velo}).
 To this end we restrict the class of considered systems \eqref{eq:sys}.

\begin{ass}[Vector relative degree] \label{ass:vec_rel} \quad 
System \eqref{eq:sys} has a square input-output structure, i.e., $\dim u = n_u = n_y = \dim y$ holds, and it has a well-defined vector relative degree
\begin{equation} \label{eq:vec_rel}
r = (r_1, \dots, r_{n_y})^T,\quad \hat r = \max\{r_1, \dots, r_{n_y}\},  \quad \rho = \sum_i^{n_y} r_i
\end{equation}
 on a sufficiently large set $\tilde{\mcl{X}} \subseteq \mcl{X}\subseteq\Rnx$.
\end{ass}

Readers not familiar with the notion of a vector relative degree of a nonlinear system are referred to \cite[Chap. 5]{Isidori95a} and \cite{Nijmeijer90a}.
In essence, the last assumption implies that \eqref{eq:sys} is locally static input-output feedback linearizable. A key ingredient in the further investigations will be the notion of a \textit{transverse normal form} of a path-following problem, which is in essence a nonlinear input-output normal form tailored to path-following problems, see also \cite{Banaszuk95a, Nielsen08a}. 
\textit{
\begin{lemm}[Local existence of a transverse normal form] \label{lem:TNF} 
Consider system \eqref{eq:sys}. Suppose that Assumption \ref{ass:vec_rel} holds and that in the timing law \eqref{eq:timing} $\hat r$ from \eqref{eq:vec_rel} is used. Then the following statements hold for all $(x, z)^T \in \mcl{N}_{\tilde x}\times\mcl{Z}$ with $\tilde x \in \interior(\tilde{\mcl{X}})$:
\begin{itemize}
\item[i)]  The augmented system \eqref{eq:sys_aug} has a well-defined vector relative degree $\tilde r = (r_1, \dots, r_{n_y}, \hat r)^T$.
\item[ii)] There exists a local diffeomorphism $\Phi: \mbb{R}^{n_x}\times\mbb{R}^{\hat r} \to \mbb{R}^{\rho}\times\mbb{R}^{n_x+\hat r - \rho}, \, (x,z) \mapsto (\xi, \eta)$ such that \eqref{eq:sys_aug} is equivalent to 
 a transverse normal form
			\begin{subequations} \label{eq:sys_transverse}
			\begin{align}
			\dot{\xi}_i  &=  \tilde{{I}}^{r_i-1} \xi_i + 
			  \begin{pmatrix}
			  {0}^{r_i-1, 1}  \\ \alpha_i(\xi_1, \dots, \xi_{n_y}, \eta, u, v)
		    \end{pmatrix}, \quad   i \in \{1, \dots, n_y\}\\
			\dot \eta &= \beta(\xi, \eta, u, v)
			\end{align}
			\end{subequations}
		with
		\[\xi = \begin{array}{c c c c c}
		\big(\underbrace{e_1, \dot e_1, \dots, e_1^{(r_1-1)}},&
		\dots, &\underbrace{e_{n_y}, \dots,
		e_{n_y}^{(r_{n_y}-1)}} & \big)^T, \\
		  \xi_1 & 
		 & \xi_{n_y}&
		\end{array} \]
  		 whereby $  \xi \in \mbb{R}^{\rho}$ and $\rho = \sum_{i=1}^{n_y} r_i$, $\eta \in \mbb{R}^{n_x +\hat r - \rho}$.
\end{itemize}
\end{lemm}
}
\begin{proof}
The proof mainly exploits the fact that the dynamics of $x$ and $z$ are only coupled via the output of \eqref{eq:sys_aug}. 
We show how the Lie derivatives of the output of the augmented system \eqref{eq:sys_aug} can be obtained, and thereby  we proof part i) of the Lemma. 
Part ii) follows directly by results given in \cite{Isidori95a, Nijmeijer90a}. 

Using the simple change of coordinates $\chi = (x,z)^T, ~\nu =(u,v)^T, ~\mu = (e,\theta)^T$ system \eqref{eq:sys_aug} can be
written as
\begin{subequations} \label{eq:sys_aug_chi}
  \begin{align}
    \dot \chi &= \phi(\chi) + \sum_{j=1}^{n_u+1} \omega_j(\chi)\nu_j \\
    \mu &= \psi(\chi).
  \end{align}
\end{subequations}
The vector fields $\phi:\mbb{R}^{n_x+\hat r} \to \mbb{R}^{n_x+\hat r}, \omega_j: \mbb{R}^{n_x+\hat r} \to \mbb{R}^{n_x+\hat r},
\psi: \mbb{R}^{n_x+\hat r}\to \mbb{R}^{n_y+ 1}$ follow directly from \eqref{eq:sys_aug}.

Calculating the Lie derivatives of $\psi_i(\chi) = h_i(x) - p_i(z_1),\, i\in\{1,\dots, n_y\}$ with respect to $\nu_j, j \in \{1, \dots, n_y+1\}$ yields
\begin{subequations} \label{eq:decoupMatElm}
\begin{equation}
  \LomegaLphipsi{j}{k}{i} =
\left\{\begin{array}{l l}
\Lie_{g_j}\Lfhx{k}{i} & j \in \{1,\dots, n_y\}\\
\Lie_{E}^{\phantom{i}}\Lie_{\tilde I}^kp_i(z_1) & j = n_y+1
\end{array}\right.
. 
\end{equation}
Assumption \ref{ass:vec_rel} 
implies that $\Lie_{g_j}\Lfhx{k}{i} =0$ for $k \in \{1,\dots, r_i-2\}, ~ i,j \in \{1,\dots, n_y\}$.
From \eqref{eq:sys_aug} it follows that 
$ \Lie_{E}^{\phantom{i}}\Lie_{\tilde I}^kp_i(z_1)=0$ for $k \in \{1,\dots, \hat r-2\},~   i \in \{1,\dots, n_y\} $.

Due to Assumption \ref{ass:vec_rel}  it is clear that for $k = r_i-1$ and at least one $j \in \{1, \dots, n_y\}$
\begin{equation}
 \LomegaLphipsi{j}{k}{i} =  \Lie_{g_j}\Lfhx{k}{i} \neq 0.
\end{equation}

\begin{figure*}[t]
		\begin{center}
  \includegraphics[width=0.72\textwidth]{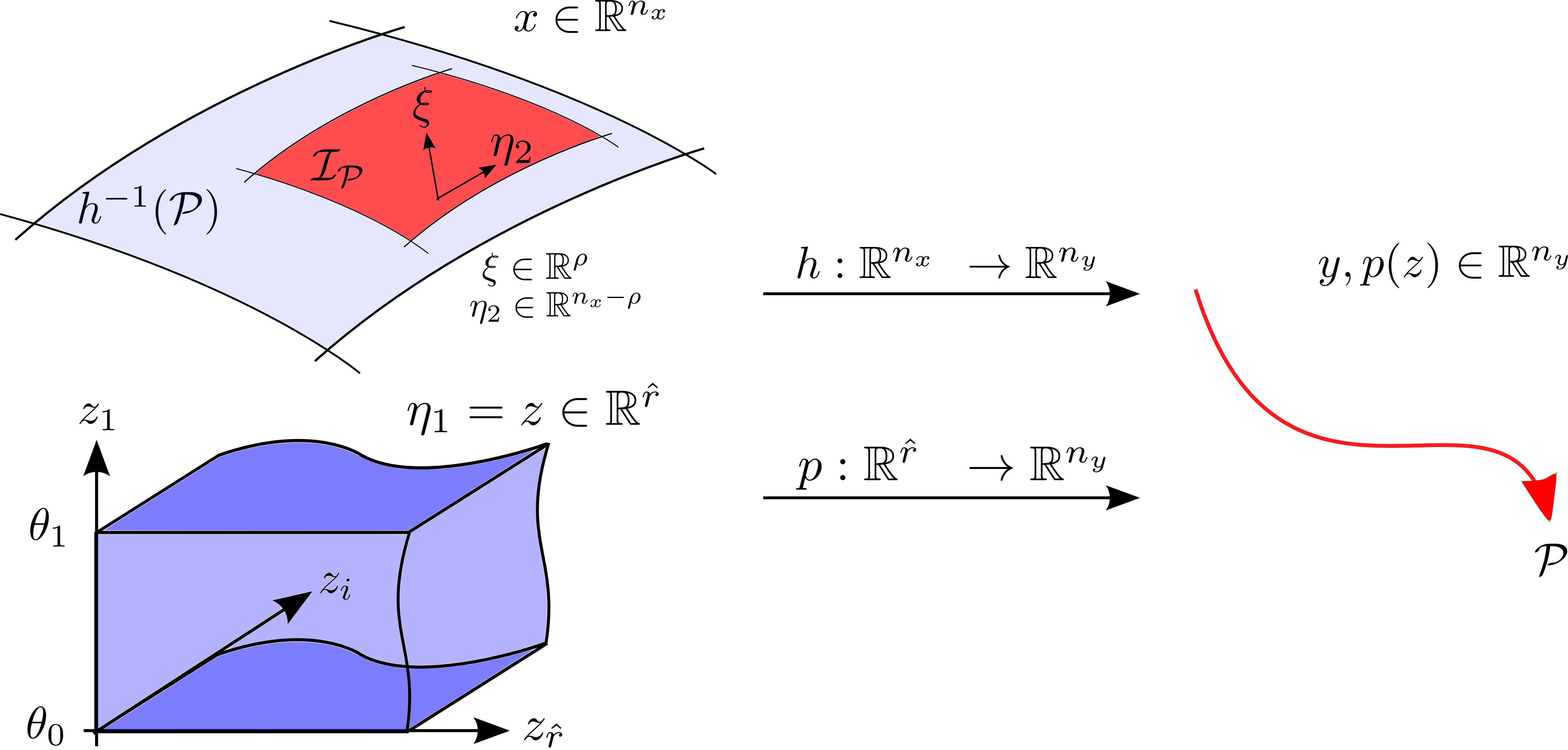}
\caption{Geometric interpretation of the transverse normal form
\eqref{eq:sys_transverse}. \label{fig:TNF}}		
\end{center} 
\end{figure*}

Now consider the case $i = n_y+1$, i.e., consider the output $\mu_{n_y+1}= \theta = z_1$ of \eqref{eq:sys_aug}. Since this output is only influenced by $\nu_{n_y+1} = v$, it follows  that
\begin{equation}
 \LomegaLphipsi{j}{k}{n_y+1} = \left\{
\begin{array}{l} 0,\, j \in \{1, \dots, n_y+1\} , k \in \{1, \dots, \hat r -2\}\\
 0,\, j \in \{1, \dots, n_y\} , k = \hat r -1 \\
1,\,  j = n_y +1,  k = \hat r -1.
\end{array}
\right.
\end{equation}
\end{subequations}
The conditions (\ref{eq:decoupMatElm}a--c) imply that  the decoupling matrix of \eqref{eq:sys_aug_chi} has the following structure
\begin{align*}
 {A}(\chi) &= 
		   \left(\begin{array}{c c c|c}
                    	   \Lie_{g_1}\Lfhx{r_1-1}{1} &\dots & \Lie_{g_{n_y}}\Lfhx{r_1-1}{1}  &  \ast \\
			   \vdots    & \ddots & \vdots &  \vdots \\
			   \Lie_{g_1}\Lfhx{r_{n_y}-1}{n_y}    &\dots & \Lie_{g_{n_y}}\Lfhx{r_{n_y}-1}{n_y} &  \ast \\  \hline
			     0 & \dots &  0 & 1
                   \end{array}\right) \\ & = \left(\begin{array}{c|c} {A}(x) & \mbf{\ast} \\ \hline  {0}^{1, n_y} & 1 \end{array}\right).	
\end{align*}
Note that this matrix is an upper triangular block matrix, whereby the decoupling matrix of the original system \eqref{eq:sys} appears as upper left block. 
The stars in the upper right block replace $\LomegaLphipsi{j}{k}{i}$ for $i \in \{1, \dots, n_y\}, j = n_y+1, k = r_i-1$. These terms  are either $0$ (for $r_i < \hat r$) or  $\neq 0$ (for $r_i = \hat r$). However, they do not affect the rank of ${A}(\chi)$. 
Due to Assumption \ref{ass:vec_rel} we know that the upper left block ${A}(x)$ of this matrix has full rank in an open neighborhood $\mcl{N}_{\tilde x}$ of $\tilde x$. Thus ${A}(\chi)$ has full rank on $\mcl{N}_{\tilde x}\times\mcl{Z}$. 
From this and \eqref{eq:decoupMatElm} it follows that part i) of the lemma is verified, i.e., on $\mcl{N}_{\tilde x}\times\mcl{Z}$ the augmented system \eqref{eq:sys_aug} has a vector relative degree of $\tilde r = (r_1, \dots, r_{n_y}, \hat r)^T$.

In order to obtain the diffeomorphism that maps the system to a (local) transverse normal form one picks $\xi_i = \Lie^{k}_\phi\psi_i$ for $k \in \{0, \dots, r_i-1\}, i = \{1, \dots, n_y\}$ as new coordinates. This specifies $\rho$ coordinates $\xi$ with $\rho = \sum_{i=1}^{n_y} r_i \leq n_x + \hat r$, which are transverse to the path manifold $\mcl{I}_\mcl{P}$ characterized by $\xi = 0$, cf. Figure \ref{fig:TNF}. The existence of additional $n_x + \hat r - \rho$ independent coordinates follows directly from the fact that the augmented system has a well-defined vector relative degree,  cf.\cite[Prop. 5.1.2]{Isidori95a} or \cite{Nijmeijer90a}.\footnote{For instance, one can pick $\hat r$ coordinates by the identity $\eta_1 = z$. This way it only remains to pick 
$n_x - \rho$ coordinates, which are directly related to the internal dynamics of \eqref{eq:sys_dyn} with respect to the output \eqref{eq:sys_out}. This situation is also illustrated in 
Figure \ref{fig:TNF}. We refer to \cite[Chap. 4]{ifat:faulwasser13a_short} for an example showing that also $\eta_1 \neq z$ might be helpful in some cases.}
 This finishes the proof. 
\end{proof}
\vspace*{0.2cm}

\begin{rema}[Transverse normal forms] \label{rem:TNF} 
Note that the directions $\xi \in \mbb{R}^\rho$ are composed of the path error $e = h(x) - p(z_1)$ and its derivatives, i.e., these directions point away from the path manifold. A graphical interpretation of this situation is depicted in Figure \ref{fig:TNF}. More precisely, these directions are transverse---i.e., orthogonal---to the manifold of trajectories, which travel along the path $\mcl{I}_\mcl{P}$. This transversality is the reason to denote \eqref{eq:sys_transverse} as a \textit{transverse normal form}.  
It should be recognized that the directions $\eta \in  \mbb{R}^{n_x +\hat r - \rho}$ are not specified. This implies that transverse normal form descriptions are usually not unique. 
\textcolor{black}{
Additionally, it is worth to be mentioned that Assumption \ref{ass:vec_rel} is only sufficient but not necessary for the existence of a transverse normal form. 
If $\dim u > \dim y$, one can use ideas from \cite[Chap. 5]{Isidori95a} to derive the normal form. If $\dim u < \dim y$, the situation is more complicated. In the special case $\dim y = \dim u +1$ one can attempt to use the virtual input $v$ to achieve $\dim y = \dim (u,v)^T$.   An example of a transverse normal form for a system with $\dim u = 1$ and $\dim y = 2$ can be found in \cite[Chap. 4.3]{ifat:faulwasser13a_short}.}

Note that such descriptions of path-following problems were initially proposed in \cite{Banaszuk95a, Nielsen08a}. Thus, results similar to Lemma \ref{lem:TNF} can, for instance, be found in \cite{Nielsen08a}. However, our approach slightly differs from these results: We work with a known path parametrization $p:\mbb{R} \to \Rny$ from \eqref{eq:path} and the corresponding augmented system \eqref{eq:sys_aug}, while the results in \cite{Nielsen08a} consider implicitly defined paths where no parametrization is known. The consideration of the path parameter states $z$ in the augmented system \eqref{eq:sys_aug} allows the description of the reference motion along of the path.
\vspace*{1mm}
\end{rema}

Beyond these structural observations the lemma reveals that output path-following implies the stabilization of a specific manifold---the so-called path manifold, denoted as $\mcl{I}_\mcl{P}$ in Figure \ref{fig:TNF}---in the state space. This manifold is locally characterized by the condition $\xi = 0$. Hence, it is not surprising that the computation of terminal regions and end penalties satisfying Theorem \ref{thm:oMPFC} is in general challenging. In essence, such a computation implies to solve at least locally a manifold stabilization problem in the presence of input and state constraints.
\footnote{At this point it is fair to ask for sufficient or necessary \textit{path-followability} conditions. In other words, one may ask for conditions ensuring that a system can be steered along a path exactly. This question is beyond the scope of this paper. Results in this direction for unconstrained and constrained systems can be found in \cite{ifat:faulwasser14b, ifat:faulwasser13a_short}.}
 One may wonder whether there is any hope to compute terminal penalties for the MPFC schemes \eqref{eq:OCP} along the lines of \cite{Chen98, Michalska93}, i.e., based on a linearization of the augmented dynamics \eqref{eq:sys_aug} around a specific point. This is in general difficult for two reasons: First, the constraints on the path parameter states $z \in \mcl{Z}$ \eqref{eq:Z} imply that the final path point $\theta_1$ is not contained in the interior of $\mcl{Z}$. Thus ellipsoidal terminal regions based on a linearization of \eqref{eq:sys_aug} or \eqref{eq:sys_transverse} at a single point of the state space are not well suited, since $(\theta_1, 0, \dots, 0)^T \in \partial \mcl{Z}$ implies that any ellipsoidal terminal region would shrink to a single point in the directions associated with the path parameter state $z$.
Second, the structure of the internal dynamics of the transverse normal form \eqref{eq:sys_transverse}  has to be taken into account. Thus it is difficult to state a general procedure for the computation of suitable terminal regions. 

\begin{rema}[MPFC without terminal constraints] \label{rem:no_term_con}
To  reduce the computational burden one might also ask for conditions which ensure path convergence without terminal constraints. For NMPC for set-point stabilization such conditions are discussed i.a. in \cite{Jadbabaie05a,Gruene09a}. For predictive path following two major issues arise if one attempts to drop the terminal constraint:
\begin{itemize}
\item[i)] The guarantees of recursive feasibility in the presence of state constraints are in general lost, respectively, rather difficult to enforce. As a remedy one could drop the state constraints of the real system \eqref{eq:sys}. The forward motion requirement ii) of Problem \ref{prob:path}, however, inevitably leads to constraints of the virtual states $z$. Thus one might need to drop the forward motion requirement as well as and merely require convergence of the path parameter, i.e., $\displaystyle \lim_{t\to \infty} \|\theta(t) - \theta_1\| = 0$. 
\item[ii)] Note that the presence of the terminal constraints, which are a compact set $\mcl{E}\subseteq \mcl{X}$, implies that all states remain bounded during the application of the MPFC scheme. This is due to the assumed continuity of solutions (Assumption \ref{ass:cont_sol}) and the fact that at the end of each prediction over a finite horizon the augmented state $(x,z)^T$ has to be inside the compact terminal set.  If neither (compact) state constraints $x \in \mcl{X}$ nor a  compact terminal constraint are considered, extra care has to be taken in order to ensure boundedness of the states. Taking Lemma \ref{lem:TNF} into account, it is clear that the states contained in the zero dynamics of the augmented system \eqref{eq:sys_aug} with respect to the outputs $e,\theta$ might cause difficulties,  for instance, due to non-minimum phase behavior. 
Preliminary results presented in \cite{ifat:faulwasser12a} indicate that via structural assumptions on the system dynamics these issues might be avoided.
\end{itemize}
\end{rema}

\begin{rema}[Generalized cost functions for MPFC] \label{rem:cost}
In the view of  the transverse normal forms of Lemma \ref{lem:TNF} one could as well penalize not only the path-following error $e$ but also its time derivatives in the cost function  $F$ as for instance considered in \cite{epfl:faulwasser13b, Boeck13a, ifat:faulwasser12a}. However, note that to enforce path convergence it suffices to rely on Assumption \ref{ass:F}, i.e., lower boundedness of $F$ by $\underline{\psi}(\|e, \theta-\theta_1 \|)$.
\end{rema}

Finally, it should be mentioned that the rewriting the augmented system \eqref{eq:sys_aug} in a transverse normal form is not necessary to design an MPFC scheme. Indeed for many examples system descriptions in transverse coordinates exist only locally. However, transverse coordinates are very helpful
in the sense that they allow to gain insight to the geometry of path-following problems and their use often simplifies the design of terminal regions. In the next section and in Appendix \ref{app:term_region} we draw upon an example from robotics to demonstrate this.

%% file: example.tex
\section{Example: Fully Actuated Robot}  \label{sec:example}

To illustrate the proposed MPFC scheme we consider a fully actuated planar robot with two degrees of freedom.
Without friction and external contact forces the dynamics of such a robot are given by
\begin{subequations}\label{eq:robot}
\begin{align} 
\begin{pmatrix} 
\dot x_1 \\ \dot x_2
\end{pmatrix}
&=
\begin{pmatrix}
x_2 \\
B^{-1}(x_1)\left( u - C(x_1, x_2)x_2 - g(x_1) \right)
\end{pmatrix} \label{eq:robot_dyn}\\
y\phantom{~~~} &= x_1  \label{eq:robot_out}  \\
y_{ca~} & = h_{ca}(x_1)
\end{align}
\end{subequations}

\begin{figure*}[t]
\centering   
\subfloat[Simulated closed-loop trajectories for 2-DoF robot.]{\label{fig:robot_1}\includegraphics[width=.99\textwidth]{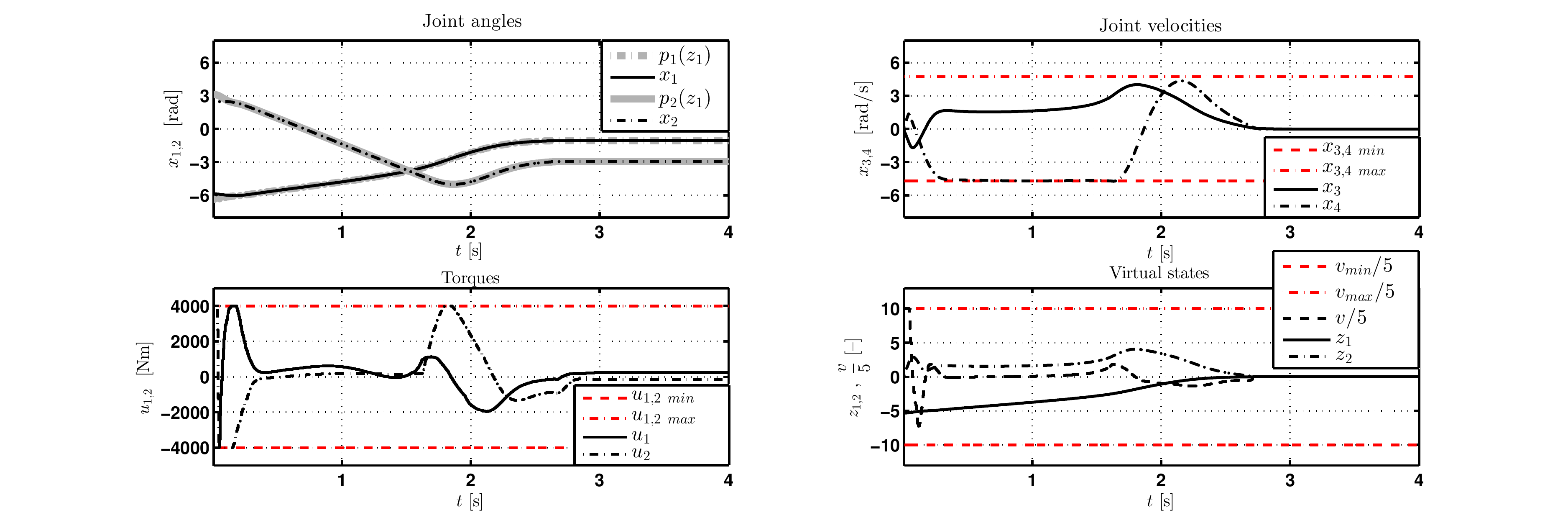}}\\
\subfloat[Path convergence in joint space (left) and Cartesian output space (right).]{ \label{fig:robot_2} \includegraphics[width=0.65\textwidth]{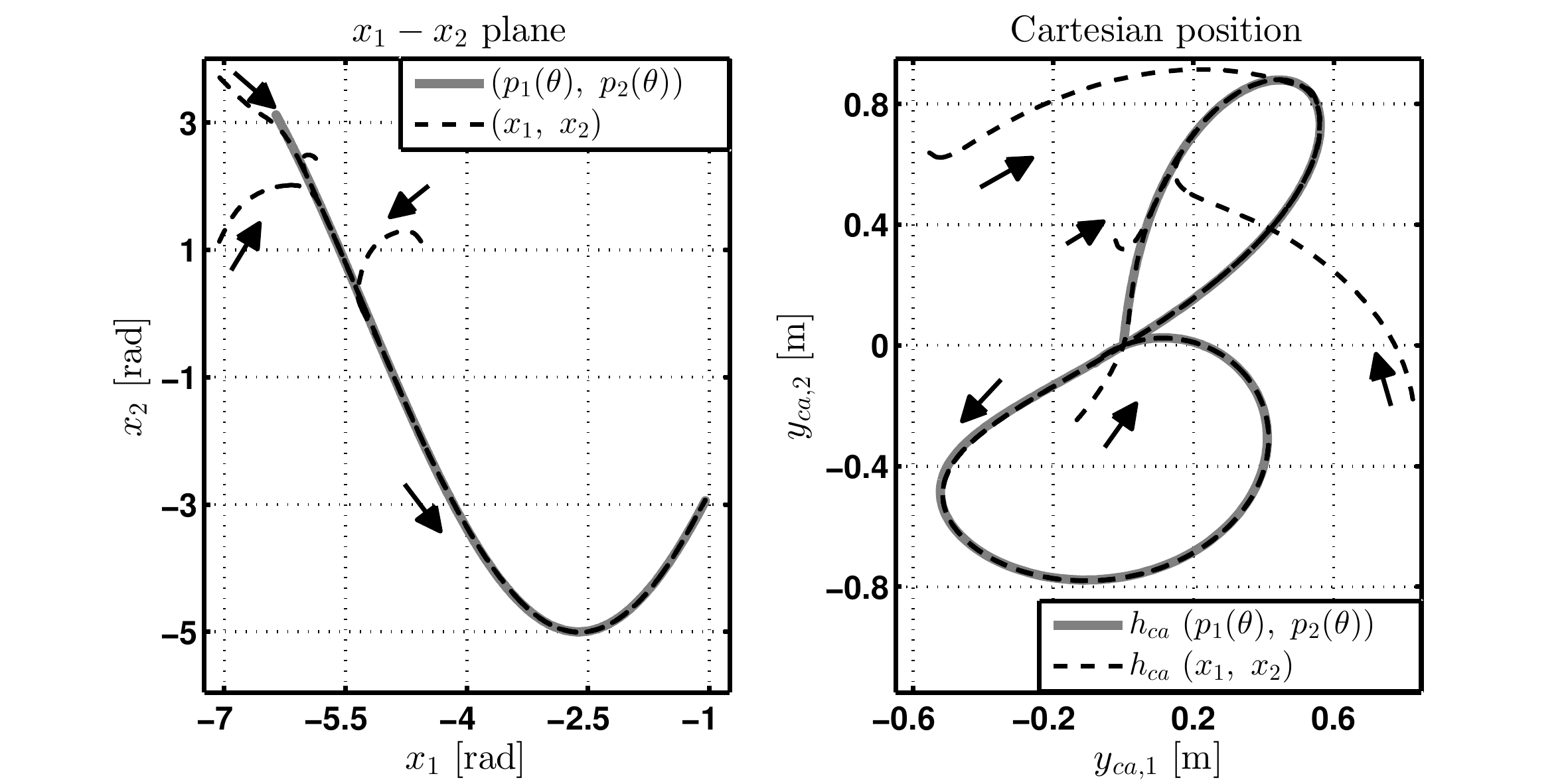}}
\caption{Simulation results for 2-DoF robot.} \label{fig:robot_sim}
\end{figure*}


Here $x_1 = (q_1, q_2) \in \mbb{R}^2$ is the vector of joint angles, $x_2 = (\dot q_1, \dot q_2) \in \mbb{R}^2$ is the vector of joint velocities.
$B: \mbb{R}^{2} \to \mbb{R}^{2 \times 2}$ and $C: \mbb{R}^{4} \to \mbb{R}^{2 \times 2}$ describe the dependence of the 
inertia on the joint angles and the dependence of centrifugal and Coriolis forces on joint angles and velocities, respectively. The function 
$g: \mbb{R}^{2} \to \mbb{R}^{2}$ models the effect of gravity. The model details are provided in Appendix \ref{app:term_region}.

The output $y=x_1$ denotes the space of joint angles, the output $y_{ca}  = h_{ca}(x_1)$ is the position of the robot tool in Cartesian coordinates. 
The inputs $u= (u_1, u_2)^T$ are the torques applied to each joint. We consider box constraints on states and inputs 
\begin{subequations} \label{eq:robot_con}
\begin{align} 
 \mcl{U} &= \left\{u \in \mbb{R}^2 ~|~ \|u\|_\infty \leq \bar u\right\} \\
\mcl{X} &= \left\{x = (x_1, x_2) \in \mbb{R}^4 ~|~ \|x_2\|_\infty = \|(\dot q_1, \dot q_2) \|_\infty \leq \bar{\dot q}\right\} 
\end{align}
\end{subequations}
whereby  $\bar u =  4000$ Nm and $\bar{\dot q} = \frac{3}{2}\pi$ rad/s.

The considered path-following task is described in the joint space. 
The path is specified via the parametrization $p: [\theta_0, \theta_1] \to \Rny$
\begin{equation} \label{eq:robot_path}
	p(\theta) = \begin{pmatrix} \theta - \frac{\pi}{3}, ~
		\omega_1 \sin(\omega_2(\theta - \frac{\pi}{3})) \end{pmatrix}^T		
\end{equation}
where $\theta_0 = -5.3, \theta_1 =0, \omega_1 = 5, \omega_2 = 0.6$.

\subsection{Simulation Results}

The cost function for the MPFC scheme is chosen according to Remark \ref{rem:cost}, i.e., we penalize the path error $e$ and it stime derivative $\dot e$. 
\begin{align} \label{eq:cost_robot}
 F(e, \dot e, \theta, u, v) & = \left \|(e, ~\dot e, ~ \theta)^T\right\|^2_Q +\left\|(u - \tilde u, ~v)^T\right\|^2_R  
\end{align}
whereby $Q = \diag(10^5, 10^5, 10, 10, 5)$ and $R = \diag(10^{-3}, 10^{-3},  10^{-4})$. This way we also satisfy Assumption \ref{ass:F}.
The offset $\tilde u = (263.0, -262.5)^T = g(p(0))$ corresponds to the torque required to keep the robot at the final path point $p(0)$.
In Appendix \ref{app:term_region} we show how to derive the following terminal region for the augmented system \eqref{eq:robot_aug}
\begin{subequations} \label{eq:robot_E}
\begin{align} 
  \mcl{E} &= \left\{(x,z) \in \mbb{R}^6 \,|\,  (\xi, \eta) = \Phi(x,z) , ~\xi^T P_\xi  \xi
  \leq 3.13, ~ \eta \in \mcl{E}_\eta \right\}
\label{eq:robot_E_x}\\
 \mcl{E}_\eta &= \big\{z \in \mbb{R}^2 \,|\, 
z_1 \in [-5.3, 0], z_2 \in [0, 0.4],\,
n^Tz \leq 0
\big\} 
\end{align}
\end{subequations}
with $n = (0.78, 0.63)^T$.
The virtual states $z$ are restricted to a polyhedral terminal region $\mcl{E}_\eta$ which is sketched in Figure \ref{fig:E_eta} in Appendix \ref{app:term_region}. Additionally, in \eqref{eq:robot_E_x} the directions of the augmented state $(x,z)$ that are transverse to the path manifold  are restricted
to an ellipsoidal terminal region, whereby $P_\xi$ is from \eqref{eq:P_xi}.
Furthermore, we show in Appendix \ref{app:term_region} that $\mcl{E}$ \eqref{eq:robot_E} satisfies the conditions of Lemma \ref{lem:ExistE}.
Thus, according to Proposition \ref{prop:E(t)}  we consider the trivial terminal penalty $E(t,x,z) =0$.

\textcolor{black}{
The simulations are performed with the following parameters:
The virtual input $v$ is restricted to $\mcl{V} = [-50, 50]$. The prediction horizon is set to $T = 0.75s$, the sampling time is $\delta = 0.005s$ and OCP \eqref{eq:OCP} is solved repeatedly 
with a direct multiple shooting implementation using 20 shooting intervals \cite{Houska11a}. }

Figure \ref{fig:robot_1} presents simulations results for the initial condition 
$x(0) = (-5.86, 2.43, 0, 0)^T$, $z(0) = (-5.3, 0)^T$.
The upper left side shows the time evolution of the joints $x_1(t) = (q_1(t), q_2(t))$ in black color and the reference $p(z_1(t))$ in gray color. The joint positions converge rapidly to the reference. The upper right side depicts the corresponding joint velocities and their constraints. In the lower right side  the virtual states $z_1 = \eta_1, z_2 = \eta_2$ and the virtual input $v$ are plotted.
One can observe that the path parameter moves forward to the end of the path at $\theta=z_1 =0$. Also note that the MPFC scheme uses the virtual input $v$ to adjust the speed along the reference.
The input torques are shown in the lower left side of Figure \ref{fig:robot_1}. Both inputs satisfy the constraints.

In Figure \ref{fig:robot_2}
 the path convergence for different initial conditions is depicted. On the left side the
 plane of joint angles $x_1 = (q_1, q_2)$ is plotted. The black arrows indicate the direction of movement of the robot.
On the right side it is shown how the solutions for different initial conditions converge to the image of the path in the Cartesian output space defined via \eqref{eq:robot_cart_output}. One can see that the proposed MPFC scheme ensures path convergence for a range of initial conditions. 
Finally, we conclude that the conditions of Theorem \ref{thm:oMPFC} can be used to design predictive path-following controllers.  In presence of constraints on states and inputs the MPFC scheme enforces path convergence and convergence on the path.

%% file: conclusion.tex
\section{Conclusions}
This paper has presented a predictive control scheme for constrained path-following problems with and without velocity assignment that guarantees convergence subject to sufficient convergence conditions based on terminal regions and end penalties. In contrast to geometric or backstepping approaches to path following the proposed  model predictive path-following control scheme allows to handle constraints on states and inputs as well as nonlinear dynamics and reference paths. Furthermore, we have established structural insights into path-following problems via transverse normal forms, which allow simplified computation of terminal regions and end penalties.

%% file: appendix.tex
\section{Computation of a Terminal Region for the Example} \label{app:term_region}

\subsection{Model Details for the Robot Example} \label{app:model_data}
The terms $B, C, g, h_{ca}$ of \eqref{eq:robot} are as follows
\begin{subequations}
\begin{align}
B(q)\phantom{, \dot q} &= \begin{pmatrix}b_1+ b_2\cos(q_2) & b_3 + b_4cos(q_2) \\ b_3 + b_4cos(q_2) & b_5 \end{pmatrix} \label{eq:robot_B} \\
C(q, \dot q) &= -c_1\sin(q_2)\begin{pmatrix}\phantom{-}\dot q_1 & \dot q_1 + \dot q_2 \\ -\dot q_1 &0 \end{pmatrix}  \label{eq:robot_C}
\\
g(q)\phantom{, \dot q} &= \begin{pmatrix}g_1\cos(q_1) + g_2 \cos(q_1+q_2), & g_2\cos(q_1+q_2) \end{pmatrix}^T   \label{eq:robot_g}\\
h_{ca}(q)& = \begin{pmatrix}l_1\cos(q_1) + l_2\cos(q_1+q_2)\\ l_1\sin(q_1) + l_2\sin(q_1+q_2) \label{eq:robot_cart_output}\end{pmatrix}.
\end{align}
\end{subequations}
The system parameters are listed in Table \ref{tab:robot}, cf. \cite{Siciliano09}.

\begin{table}[b]
 \begin{center}
\caption{Robot parameters \cite{Siciliano09}. \label{tab:robot}}
 \begin{tabular}{|l r l || l r l |} \hline 
  $b_1\phantom{\big(}$	&  $\phantom{-}200.0$ & $ [$kg m$^{2}/$rad$]$  & $b_{2}\phantom{\big(}$ 	&  $\phantom{-}50.0$ & $ [$kg m$^{2}/$rad$]$ \\
  $b_3$	&  $\phantom{-0}23.5 $&$  [$kg m$^{2}/$rad$]$ & $b_{4}$ 	&  $\phantom{-}25.0$ & $ [$kg m$^{2}/$rad$]$   \\
  $b_5$	&  $\phantom{-}122.5 $&$ [$kg m$^{2}/$rad$]$  & $c_{1}$ 	&  $-25.0 $&$ [$Nms$^{-2}]$   \\
  $g_1$	&  $\phantom{-}784.8 $&$ [$Nm$]$& $g_{2}$ 	&  $\phantom{-}245.3 $&$ [$Nm$]$   \\
  $l_1$ &  $\phantom{-00}0.5 $&$ [$m$]$& $l_{2}$ 	&  $\phantom{-00}0.5 $&$ [$m$]$ \\
\hline
  \end{tabular}
 \end{center}
 \end{table}

\subsection{Problem Description in Transverse Normal Form}
It is easy to see that \eqref{eq:robot} has a global vector relative degree $r = (2,2)^T$ with respect to the output $y =x_1$.
Thus, we use as path parameter dynamics an integrator chain of length two and
obtain the augmented system description 
\begin{subequations} \label{eq:robot_aug}
\begin{align}
\begin{pmatrix} 
\dot x_1 \\ \dot x_2 \\ \dot z
\end{pmatrix}
&=
\begin{pmatrix}
x_2 \\
B^{-1}(x_1)\left( u - C(x_1, x_2)x_2 - g(x_1) \right) \\
\tilde{{I}}^{2}z + {E}^2v
\end{pmatrix} \\
e & = x_1 - p(z_1) \\
\theta & = z_1.
\end{align}
\end{subequations}
We want to map these augmented dynamics into a transverse normal form.
Following along the lines of the proof of Lemma \ref{lem:TNF} we obtain
the coordinate transformation $\Phi:~ \mbb{R}^4\times\mbb{R}^2 \to \mbb{R}^4\times\mbb{R}^2$ and its inverse $\Phi^{-1}:~\mbb{R}^4\times\mbb{R}^2 \to \mbb{R}^4\times\mbb{R}^2$
\begin{subequations} \label{eq:Phi}
\begin{align} 
 \Phi\phantom{^{-1}}: \hspace*{0.5mm} \xi_1 &=  x_1 - p(z_1) , \quad \xi_2 =   x_2 - \frac{\partial p}{\partial z_1} z_2, \quad  \eta = z  \\
\Phi^{-1}:  x_1 &=  \xi_1 + p(\eta_1) , \quad x_2 =   \xi_2 + \frac{\partial p}{\partial \eta_1} \eta_2, \quad  z = \eta.
\end{align}
\end{subequations}
Observe that to simplify the later derivations,  we have chosen a different ordering of the $\xi$-variables compared to \eqref{eq:sys_transverse}. Furthermore, note that only the virtual states $z$ appear in $\eta$. This is due to the fact that \eqref{eq:robot_dyn} has a state dimension of $n_x =4$ and the 
vector relative degree \eqref{eq:robot_dyn} $r = (2,2)^T$ with respect to 
$y =x_1$.  Thus the robot dynamics \eqref{eq:robot_dyn} do not have internal dynamics with respect to \eqref{eq:robot_out}.

Due to its simple structure and the assumptions on the path parametrization $p(\theta)$ it is easy to see that $\Phi:~ \mbb{R}^4\times\mbb{R}^2 \to \mbb{R}^4\times\mbb{R}^2$ is a global diffeomorphism.
Using $\Phi$ it is straightforward to rewrite the augmented robot dynamics \eqref{eq:robot_aug} into the transverse normal form.
We obtain
\begin{subequations} \label{eq:robot_TNF}
\begin{align}
\begin{pmatrix} 
\dot \xi_1 \\ \dot \xi_2 \\ \dot\eta
\end{pmatrix}
&=
\begin{pmatrix}
\xi_2 \\ 
\alpha(\xi,\eta, u, v) \\
\tilde{{I}}^{2} \eta + {E}^2v
\end{pmatrix} \\
  e & = \xi_1 - p(\eta_1) \\
 \theta & = \eta_1,
\end{align}
\end{subequations}
whereby the vector field $\alpha: \mbb{R}^4\times\mbb{R}^2\times\mbb{R}^2\times\mbb{R} \to \mbb{R}^2$ is
\begin{multline} \label{eq:alpha_robot_tnf}
\alpha(\xi,\eta, u, v) = 
 B^{-1}(\xi_1, \eta_1)\Bigg( u - C(\xi, \eta)\left(\xi_2 - \frac{\partial p}{\partial \eta_1} \eta_2 \right) \\ - g(\xi_1, \eta_1) \Bigg) 
- \frac{\partial^2 p}{\partial \eta_1^2} (\eta_2)^2 - \frac{\partial p}{\partial \eta_1} v.
\end{multline}

\subsection{Design of a Terminal Region}
We design a terminal control laws for the augmented dynamics in transverse normal form \eqref{eq:robot_TNF}.
Note that in \eqref{eq:robot_aug} as well as in \eqref{eq:robot_TNF} the dynamics of $z = \eta$ are not influenced by the other states. Thus we first design a terminal control and terminal region for the $\eta$-dynamics and subsequently consider the transverse dynamics. 
Recall that the path parameter dynamics  $\dot\eta = \tilde{{I}}^{2} \eta + {E}^2 v$ are simply a double integrator and the state constraint $\mcl{Z}$  \eqref{eq:Z} is a polytope 
$ \mcl{Z} = \{\eta\in\mbb{R}^2 ~|~ \eta_1 \in [\theta_0, 0], ~\eta_2 \geq 0\}$ 
with $\eta =z$. A sketch of the state constraint $\mcl{Z}$ is shown in Figure \ref{fig:E_eta}.  Part ii) of Problem \ref{prob:path} requires that  a terminal control law achieves $\lim_{t\to\infty}\eta(t) = \lim_{t\to \infty} z(t) = (0, 0)$.
In other words, the path parameter state $\eta$ should converge to the origin and the constraint $\eta(t) \in \mcl{Z}$ implies that $\theta =\eta_1 \in [\theta_0, 0]$ and $\dot\theta=\eta_2 \geq
 0$. Since the origin is contained in the boundary of the $\mcl{Z}$, ellipsoidal terminal regions for the $\eta$ part of  \eqref{eq:robot_TNF} would shrink to a point.  Thus we aim on constructing a 
polytopic terminal region, which is rendered positively invariant by a linear feedback  $v_\mcl{E} = K\eta$.

The coefficients of  $K_\eta$ should satisfy $k_1, k_2 < 0$ in order enforce asymptotic convergence to the origin and $~k_2^2 > -4k_1$ to avoid oscillations. 
Furthermore, it is easy to verify that the eigenspaces of  $\tilde{{I}}^{2} +{E}^2 K_\eta$ corresponding to such a choice for $K_\eta$ lie in the $2^{nd}$ and $4^{th}$ quadrant of the $\eta_1-\eta_2$ phase plane. Exemplarily this is depicted by the blue lines in Figure \ref{fig:E_eta}.
We use 
\begin{multline} \label{eq:kz}
 v_\mcl{E} = K_\eta \eta, \quad K_\eta = (k_1, k_2), \\~k_1, k_2 <0,  ~k_2^2 > -4k_1, ~k_2 \leq -k_1\theta_0\left(\bar{\dot\theta}\right)^{-1}< 0
\end{multline}
as a terminal feedback for $\eta$. The additional condition $k_2 \leq -k_1\theta_0/\bar{\dot\theta}$ 
 ensures that the initial velocity vector of any closed-loop solution starting on the line $\eta = \left(\alpha, \bar{\dot\theta}\right), \alpha \in [\theta_0, 0]$ points towards the $\eta_1$-axis.
Furthermore,  it is easy to verify that for such a choice of $K_\eta$ all solutions starting somewhere in the interval $[\theta_0, 0]$ on the negative $\eta_1$ axis converge to the origin with $\dot \theta = \eta_2 >0$. Additionally, we have for any choice of $K_\eta$ that the solutions starting on the positive $\dot\theta=\eta_2$ axis will leave the state constraint set $\mcl{Z}$.
Based on these considerations we choose the terminal region for the $\eta$-dynamics as 
\begin{equation} \label{eq:E_eta}
 \mcl{E}_\eta := \left\{\eta \in \mbb{R}^2 ~|~ \eta_1 \in [\theta_0, 0], ~ \eta_2 \in \left[0, \bar{\dot\theta}\right], ~ n_1^T\eta \leq 0\right\}.
\end{equation}
 This terminal region is sketched in green color in Figure \ref{fig:E_eta}. 
Here, $n_1$ is the normal vector corresponding to the \textit{upper} eigenspace of $\tilde{{I}} ^{2}+E^2 K_\eta$.
The verification of  positive invariance of $\mcl{E}_\eta$ with respect to $\dot \eta = (\tilde{{I}} ^{2}+{E}^2 K_\eta)\eta$ follows directly from the considerations above. 
The boundary $0\leq  \eta_2 \leq \bar{\dot\theta}$
is introduced to $\mcl{E}_\eta$ in order to simplify the design of a terminal region for the $\xi$-dynamics. 

 \begin{figure}[t]
	\begin{center}
     	\includegraphics[width =.3\textwidth]{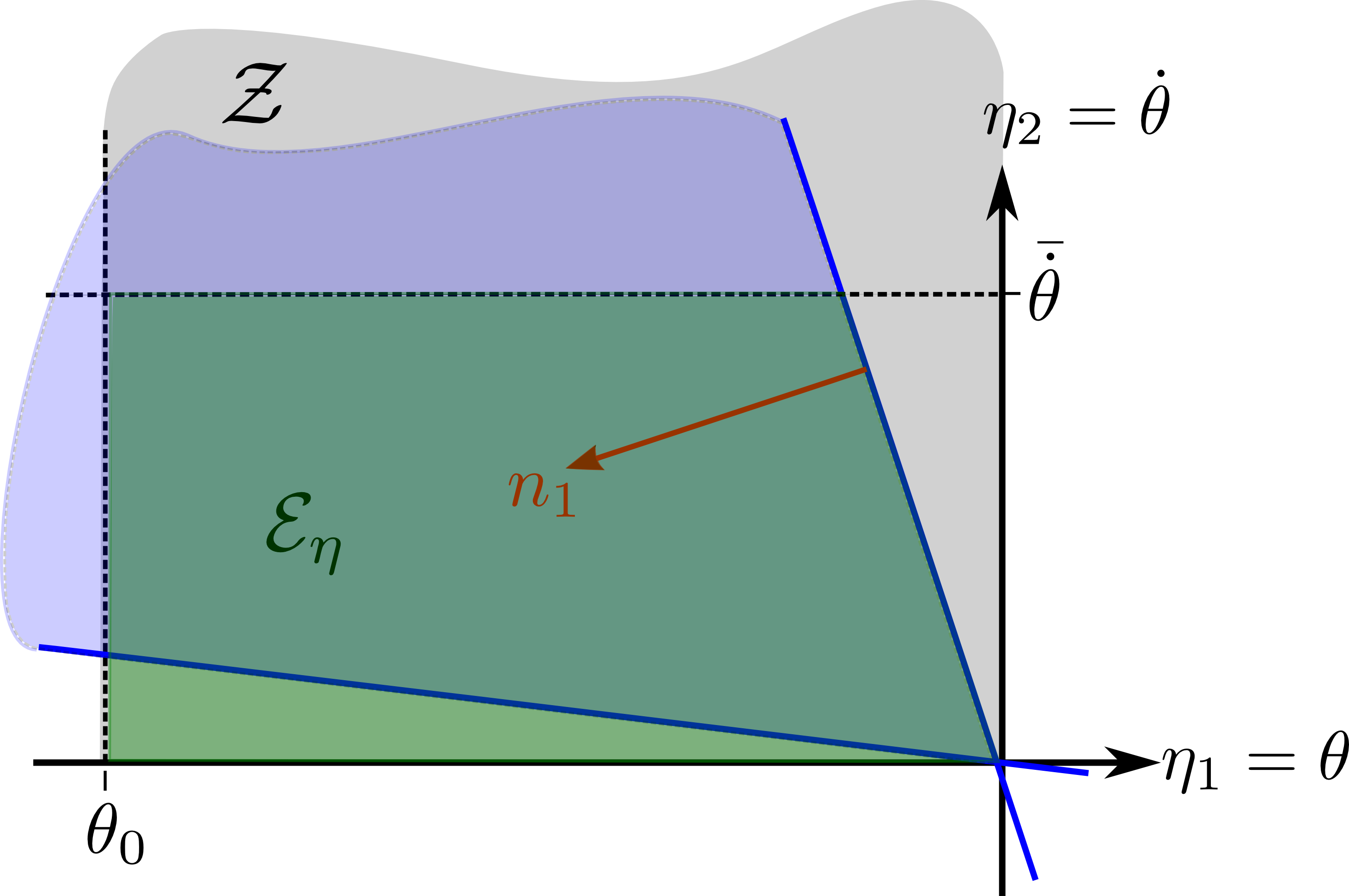}
	\caption{Terminal region for the path parameter dynamics.  \label{fig:E_eta}}
	\end{center}
 \end{figure}

We proceed with the design of a suitable feedback for the $\xi$-part of the transverse dynamics \eqref{eq:robot_TNF}.
As a terminal feedback we use  
\begin{multline} \label{eq:u_E_x}
 u_\mcl{E}(\xi, \eta) =C(\xi, \eta)\left(\xi_2-\frac{\partial p}{\partial \eta_1}\eta_2\right) + g(\xi_1, \eta_1) \\
+ B(\xi_1, \eta_1)\left(K_\xi\xi + \ddot p(\eta_1(t))  \right).
\end{multline}
It is easy see to that this feedback achieves global exact static feedback linearization of \eqref{eq:robot_TNF}. More precisely it achieves global transverse feedback linearization, cf. \cite{Banaszuk95a, Nielsen08a}.
The term $K_\xi\xi$ will be used to stabilize the path manifold. 
The $\ddot p(\eta_1(t)$ part can be understood as a feedforward control. Using this feedback the $\xi$-part of \eqref{eq:robot_TNF}
is governed by $\dot \xi = (A_\xi + B_\xi K_\xi)\xi$.
W.l.o.g. we assume that we have designed a stabilizing
gain matrix $K_\xi$ and that 
\begin{equation} \label{eq:V_xi}
V(\xi) = \xi^TP_\xi\xi, \quad P_\xi >0
\end{equation}
 is a corresponding Lyapunov function.

Now, we are ready to derive a terminal region $\mcl{E}_\xi\subset \mbb{R}^4$ for the transverse part of \eqref{eq:robot_TNF}.  
The main idea is to bound the norm of the feedback \eqref{eq:u_E_x} from above and to obtain the terminal region $\mcl{E}_\xi\subset \mbb{R}^4$ as a level set of $V(\xi)$. 
Due their structure the terms $B: \mbb{R}^2 \to \mbb{R}^{2\times2}, C:\mbb{R}^2 \times \mbb{R}^2 \to \mbb{R}^{2\times2}$ and $g:\mbb{R}^2 \to \mbb{R}^2$ from (\ref{eq:robot_B}-f) can be bounded from above by constants
\begin{equation*}
 \forall x \in \mcl{X}: \quad\|B(x_1)\| \leq \bar B, \quad  \|C(x_1, x_2)\| \leq \bar C, \quad \|g(x_1)\| \leq \bar g. 
\end{equation*}
These bounds also hold in in $(\xi, \eta)$ coordinates. 
To bound $\|\ddot p(z_1(t))\|$ from above we restrict ourselves to the set $\mcl{E}_\eta$ from \eqref{eq:E_eta}. 
Since $z = \eta$,  $p(\theta) \in \mcl{C}^2$ and $\mcl{E}_\eta$ is compact, we obtain 
\begin{multline*}
 \forall \eta \in \mcl{E}_\eta: \quad \left\|\ddot p(\eta_1(t))\right\| \leq  \left\| \frac{\partial^2 p}{\partial \eta_1^2}\eta_2^2 + \frac{\partial p}{\partial \eta_1}v\right\| \\
\leq \left\| \frac{\partial^2 p}{\partial \eta_1^2}\right\|\left(\bar{\dot\theta}\right)^2 + \left\|\frac{\partial p}{\partial \eta_1}\right\|
 \left\|k_1\theta_0 + k_2\bar{\dot\theta}\right\| =: \bar{\ddot p}.
\end{multline*}
Here, we have used that in the set $\mcl{E}_\eta$ the $\eta$-dynamics are controlled via  $v = K_\eta \eta$.
To simplify the further considerations we work with tightened constraints  $\bar{\mcl{U}} \subset \mcl{U}$, $\bar{\mcl{X}} \subset \mcl{X}$
\begin{subequations}\label{eq:robot_con_tight}
\begin{align} 
 \bar{\mcl{U}} &= \left\{u \in \mbb{R}^2 ~|~ \|u\| \leq \bar u\right\} \\
\bar{\mcl{X}} &= \left\{x = (x_1, x_2) \in \mbb{R}^4 ~|~ \|x_2\| = \|(\dot q_1, \dot q_2) \|\leq \bar{\dot q}\right\}
\end{align}
\end{subequations}
where in comparison to \eqref{eq:robot_con} the $2$-norm is used instead of
$\|\cdot\|_\infty$.

Next, we apply the bounds derived before to the feedback $u_\mcl{E}$ from \eqref{eq:u_E_x}. This yields
\[
 \forall (\xi,\eta)^T \in \Phi\left( \bar{\mcl{X}}\times \mcl{E}_\eta\right):  \quad  \| u_\mcl{E}(\xi,\eta)\| \leq \bar C\bar{\dot q} +\bar g +  \bar B\left(\bar{\ddot p} + \|K_\xi \xi\| \right).
\]
We enforce that inside the terminal region to be determined, $\mcl{E}_\xi\times\mcl{E}_\eta$, the tightened input constraint $u_\mcl{E}(\xi, \eta) \in  \bar{\mcl{U}}$ is satisfied. This is the case if 
 \[
   \forall (\xi,\eta)^T \in \mcl{E}_\xi\times\mcl{E}_\eta: \quad  \bar C\bar{\dot q} +\bar g +  \bar B\left(\bar{\ddot p} + \|K_\xi \xi\| \right) \leq \bar u.
 \]
Solving the last equation for $\|\xi\|$ yields for all $(\xi,\eta)^T \in \Phi\left( \bar{\mcl{X}}\times \mcl{E}_\eta\right)$:
\begin{multline} \label{eq:u_bnd_robot}
\|\xi\| \leq \dfrac{\bar u - \bar C\bar{\dot q} -\bar g - 
\bar B \bar{\ddot{p}}}{\bar B \|K_\xi\|} 
\quad \Rightarrow \quad u_\mcl{E}(\xi, \eta) \in\bar{\mcl{U}} \subset \mcl{U}.
\end{multline}
Subsequently,  we derive $\mcl{E}_\xi$ as a suitable level set of the Lyapunov function $V(\xi)$ from \eqref{eq:V_xi}. In general, the level set is
\[
 \mcl{E}_\xi := \left\{ \xi \in \mbb{R}^4 ~|~ \xi^TP_\xi \xi \leq \gamma^2 \right\}.
\]
The constant $\gamma$ can  be computed as follows
\begin{subequations} \label{eq:compute_c}
\begin{equation}
 \underset{\gamma > 0}\maximize ~\gamma
\end{equation}
subject to
\begin{align}
 \forall \xi \in \mcl{E}_\xi: \quad \phantom{_2} \|\xi\| &\leq \dfrac{\bar u - \bar C\bar{\dot q} -\bar g - \bar B \bar{\ddot{p}}}{\bar B \|K_\xi\|} \label{eq:u_lev_set_con}\\
 \forall  \xi \in \mcl{E}_\xi: \quad \left\| \xi_2 \right\| &\leq \bar{\dot q}- \bar{\dot p}.\label{eq:q_lev_set_con}
\end{align}
\end{subequations}
Here, $\bar{\dot p}$ is a bound on $\dot p(\eta_1(t))$ that can be obtained for $\eta \in\mcl{E}_\eta$ in a similar fashion as $\bar{\ddot p}$. 
Given $K_\xi$ and $P_\xi$ this is a simplified version of the (convex) problem to compute a maximum volume ellipsoid contained in a convex set, cf. \cite{Boyd04}.
If $\bar{\dot q}-\bar{\dot p}$ and the constant on the right side of \eqref{eq:u_lev_set_con} are positive, problem \eqref{eq:compute_c} has a solution $\gamma^\star >0$. This is the case if the input bound $\bar u$ and the bound $\bar{\dot q}$ are sufficiently large.

 We use the model data from Table \ref{tab:robot} and the path \eqref{eq:robot_path} to compute numerically the sets $\mcl{E}_\eta$ and $\mcl{E}_\xi$. The bound on $\eta_2$ is set to $\bar{\dot \theta} = 0.4$, and
the feedback matrix for the $\eta$-dynamics is $K_\eta = (-0.1, -1.33)$. 
This leads to the terminal constraint for $\eta$
 \begin{equation*}
 \mcl{E_\eta} = \left\{\eta \in \mbb{R}^2 ~|~   \eta_1 \in [-5.3, 0], ~\eta_2 \in [0, 0.4],  (0.78, \,0.63)\eta \leq 0\right\}.
\end{equation*}
The Lyapunov function \eqref{eq:V_xi} and the feedback matrix $K_\xi$ 
%
are computed via an LQR controller with $Q_\xi ={I}^{4}, R_\xi = {I}^{2}$.
This leads to 
\begin{equation} \label{eq:P_xi}
    P_\xi = \left(\begin{array}{c|c}
		P_1 & P_2 \\ \hline P_2 & P_1
\end{array}\right),  \quad K_\xi = (P_1, P_2) 
\end{equation}
with $P_1 = \diag(1.73, ~ 1.73)$ and $P_2 = {I}^2$.
Solving \eqref{eq:compute_c} with these values yields $ \gamma = 1.77$.
Thus the ellipsoidal part of the terminal region is
\begin{equation*} 
 \mcl{E}_\xi = \left\{ \xi \in \mbb{R}^4 ~|~ \xi^T P_\xi \xi \leq 3.13\right\}.
\end{equation*}
Rewriting the terminal constraints in $(x,z)$ coordinates yields
\begin{equation} \label{eq:robot_EII}
  \mcl{E} = \left\{(x,z) \in \mbb{R}^6 ~|~  (\xi, \eta) = \Phi(x,z) , ~\xi^T P_\xi  \xi
  \leq 3.13), ~ \eta \in \mcl{E}_\eta \right\}.
\end{equation}

\subsection{Derivation of a Terminal Penalty}
It remains to derive an end penalty such that the conditions of Theorem \ref{thm:oMPFC} or Proposition \ref{prop:E(t)} are satisfied.

For the MPFC controller  we use the quadratic cost function $F$ from \eqref{eq:cost_robot}, which can be written in $\xi, \eta$ coordinates as $F(\xi, \eta_1,u,v) = \|(\xi, \eta_1)^T\|_Q^2 + \|(u-\tilde u, v)^T\|_R^2$. 
%
It is straightforward to see that  for all $(\xi_0, \eta_0 )^T \in \mcl{E}, \forall t \geq t_0:$
\begin{align*}
\|\xi(t, t_0, \xi_0 | u_\mcl{E}(\cdot))\| &\leq c_\xi(\xi_0, \eta_0) e^{-\alpha_\xi (t-t_0)} \\
\|\eta(t, t_0, \eta_0 | v_\mcl{E}(\cdot))\| &\leq  c_\eta(\xi_0, \eta_0) e^{-\alpha_\eta (t-t_0)}, 
\end{align*}
whereby $ c_\xi(\xi_0, \eta_0)$ and $c_\eta(\xi_0, \eta_0)$ are bounded from above by finite numbers.
In other words, inside the terminal region \eqref{eq:robot_EII} the application of the terminal control law \eqref{eq:u_E_x} leads to exponential convergence of the transverse directions $\xi\in\mbb{R}^4$  and the path parameter state $\eta \in \mbb{R}^2$. Furthermore, it is clear that $v_\mcl{E}(t) = K_\eta\eta(t)$ is also converging exponentially to zero. 
Using these bounds on $\xi(t), \eta(t)$ and $v_\mcl{E}(t)$ we obtain that the solutions driven by the terminal feedback \eqref{eq:u_E_x} satisfy
\begin{multline*}
  \|u_\mcl{E}(\xi(t), \eta(t))-\tilde u\| \leq 
\underbrace{\|C(x_1(t), x_2(t))x_2(t)\|}_{\leq ~ \bar C  c_{x_2} e^{-\alpha_{x_2} (t-t_0)} } \\ 
+\underbrace{\|B(x_1(t))\left(K_x\xi(t) + \ddot p(\eta_1(t))  \right)\|}_{ \leq ~ \bar B\left(\|K_\eta\|c_\eta e^{-\alpha_\eta (t-t_0)} +c_{\ddot p} e^{-\alpha_{\ddot p} (t-t_0)}\right)}  + 
 \underbrace{\| g(x_1(t)) - g(p(0))\|}_{\leq ~c_g e^{-\alpha_g (t-t_0)}}
\end{multline*}
The bound on the first term follows from $x_2 = \xi_2 - \dot p$ and $\|\dot p(t)\| \leq \|\frac{\partial p}{\partial \theta}\| c_\theta e^{-\alpha_\theta (t-t_0)}$. The bound on the second term follows in a similar fashion. The estimate from above on $\|g(x_1) - g(p(0))\|$ is more complicated: Note that $x_1 = \xi + p(\eta_1)$.
For $p: [\theta_0, 0] \to \mbb{R}^2$ from \eqref{eq:robot_path} one can show that exponential convergence of $\eta$ to $0$ implies exponential convergence of $p(\eta_1)$ to $p(0)$. 
Using this we see that for $t \to \infty$ also the state $x_1$ converges exponentially to $p(0)$ since $x_1 = \xi + p(\eta_1)$. Finally, we use that in $g: \mbb{R}^2 \to \mbb{R}^2$ from \eqref{eq:robot_g}
only $\cos$-functions appear, and conclude that $g(x_1) - g(p(0))$ converges exponentially to $0$. Since all the arguments of $F$ converge exponentially to $0$ and $F$ is quadratic, 
we see that the terminal region \eqref{eq:robot_EII} and the terminal controls (\ref{eq:kz}, \ref{eq:u_E_x})  satisfy the assumptions of Lemma \ref{lem:ExistE}. In other words, $\mcl{E} $ from \eqref{eq:robot_EII} and the trivial terminal penalty $E(t,x(t),z(t)) =0$ allow applying Proposition \ref{prop:E(t)}.